\documentclass[draftcls,onecolumn]{IEEEtran}

\usepackage[utf8]{inputenc}
\usepackage[T1]{fontenc}
\usepackage{url}
\usepackage{ifthen}
\usepackage{cite}
\usepackage[cmex10]{amsmath} 
\usepackage{graphicx}
\usepackage{amsmath, amsthm, amssymb, amsfonts}
\usepackage{tabularx}
\usepackage[linesnumbered,ruled]{algorithm2e}
\usepackage{algpseudocode}
\usepackage{url}
\usepackage{xcolor}
\usepackage{multirow}
\usepackage{colortbl}
\usepackage{times}
\usepackage{array}
\usepackage{stfloats}
\usepackage{booktabs}

\newtheorem{definition}{Definition}
\newtheorem{lemma}{Lemma}

\newtheorem{theorem}{Theorem}
\newtheorem{remark}{Remark}

\newtheorem{example}{Example}
\newtheorem{claim}{Claim}

\usepackage{textcomp}

\usepackage{verbatim}
\usepackage{rotating} 
\usepackage{empheq}
\usepackage{tikz}
\usepackage{adjustbox} 
\usepackage{mathrsfs} 

\usepackage[caption=false,font=footnotesize,labelfont=rm,textfont=rm]{subfig}
\usepackage{caption}
\usepackage{pifont}

\usepackage{subcaption}
\begin{document}

\title{Grouped k-threshold random grid-based visual cryptography scheme}
\author{Xiaoli Zhuo, Xuehu Yan, and Wei Yan
\thanks{Xiaoli Zhuo, Xuehu Yan, and Wei Yan are with the College of Electronic Engineering, National University of Defense Technology, Hefei 230037, China, and also with Anhui Province Key Laboratory of Cyberspace Security Situation Awareness and Evaluation, Hefei 230037, China~(email: yan.wei2023@nudt.edu.cn). (Corresponding author: Wei Yan)}
}

\markboth{Journal of \LaTeX\ Class Files,~Vol.~14, No.~8, August~2021}%
{Shell \MakeLowercase{\textit{et al.}}: A Sample Article Using IEEEtran.cls for IEEE Journals}


\maketitle

\begin{abstract}
Visual cryptography schemes (VCSs) belong to a category 
of secret image sharing schemes that do not require cryptographic knowledge for decryption, 
instead relying directly on the human visual system. 
Among VCSs, random grid-based VCS (RGVCS) has garnered 
widespread attention as it avoids pixel expansion while requiring no basic matrices design. 
Contrast, a core metric for RGVCS, 
directly determines the visual quality of recovered images, 
rendering its optimization a critical research objective.
However, existing $(k,n)$ RGVCSs still fail to attain theoretical upper bounds on contrast, 
highlighting the urgent need for higher-contrast constructions.
In this paper, we propose a novel sharing paradigm for RGVCS that constructs $(k,n)$-threshold schemes from arbitrary $(k,n')$-threshold schemes $(k \leq n'\leq n)$, termed \emph{$n'$-grouped $(k,n)$ RGVCS}. 
This paradigm establishes hierarchical contrast characteristics:
participants within the same group achieve optimal recovery quality, 
while inter-group recovery shows a hierarchical contrast.
We further introduce a new contrast calculation formula tailored to the new paradigm. 
Then, we propose a contrast-enhanced $(k,n)$ RGVCS by setting $n'= k$,  
achieving the highest contrast value documented in the existing literature. 
Theoretical analysis and experimental results demonstrate the superiority of 
our proposed scheme in terms of contrast.  
\end{abstract}

\begin{IEEEkeywords}
Random grids, visual cryptography scheme, contrast. 
\end{IEEEkeywords}

\section{Introduction}
In the era of digital transformation, 
the security of sensitive information is critical in domains such as finance, healthcare, and military operations.
Although cryptographic techniques protect secret data, centralized storage introduces risks, including single points of failure, insider threats, and targeted attacks.
To address these challenges, secret sharing (SS) has emerged as a solution. 
This technique divides secret information into multiple parts, termed \emph{shares},
which are distributed among distinct parties.
The original secret can only be reconstructed when a minimum number of shares meeting specific conditions are combined. 
The most widely used form of access structure in SS is the $(k,n)$-threshold scheme \cite{k_n_threshold}, in which a secret is divided into $n$ shares, and reconstruction requires at least $k$ shares. 
Furthermore, the $(k,n)$-threshold structure has been extended to 
$(k,\infty)$ access structure\cite{evolving_k_threshold,tcom}, 
general access structure\cite{general_access_structure}, and ramp secret sharing\cite{ramp1,ramp2}. 

Although SS is initially designed for textual or numerical data, 
the growing demand for protecting multimedia data, such as images, has led to 
the emergence of secret image sharing (SIS).
It divides the information of a secret image into multiple shares (shadow images) 
distributed to multiple participants, and only by combining a sufficient number of qualified shadow images 
can the secret image be recovered. 
Among the mainstream SIS schemes, three approaches stand out: 
the visual cryptography scheme (VCS)\cite{vcs1,vcs2}, 
the polynomial-based SIS\cite{k_n_threshold, polynomial3}, and the CRT-based SIS\cite{CRT1}. 

Unlike the other two schemes that require computational functions for recovery, 
VCS does not need complex decryption methods, 
but instead relies on the human visual system (HIS) to recover the secret 
by stacking a certain number of shadow images.
Due to its simplicity of implementation and high intuitiveness, 
VCS is particularly suitable for use in low-tech environments or resource-constrained scenarios.  
However, VCS also has some limitations. 
For instance, 
VCS cannot achieve lossless recovery. 
Hence, enhancing the visual quality of the recovered images\cite{visual_quality1,visual_quality3} 
has become a crucial goal in VCS. 
Additionally, 
pixel expansion in traditional VCS, like basic matrix-based VCS (BMVCS)\cite{vcs1}, 
always leads to an increase in image size, 
which raises storage and transmission costs. 
Thus, many studies are dedicated to 
minimizing pixel expansion\cite{pixel_expansion_1,pixel_expansion_3} or 
proposing size-invariant VCS (SIVCS) without pixel expansion, 
including probabilistic VCS (PVCS)\cite{PVC1, PVC3} and 
random grid-based VCS (RGVCS)\cite{RGVCS1,RGVCS2}. 
  
In VCS, contrast serves as a critical metric for evaluating the visual quality of recovered images, 
where higher contrast correlates with enhanced clarity. 
Due to the reduced storage and transmission overhead associated with shadow images in SIVCS, 
researchers predominantly favor designing contrast enhancement schemes within SIVCS. 
PVCS-based approaches typically involve the design of intricate basis matrices 
and frequently employ linear programming techniques to ascertain optimal contrast\cite{pvc_contrast_enhancement1,pvc_contrast_enhancement2}, 
which significantly escalates computational demands and complexity. 
Unlike PVCS, RGVCS avoids the need for intricate probabilistic models and basic matrices designs, 
thereby facilitating more efficient and computationally less intensive contrast optimization processes. 

Early RGVCS studies\cite{2011_chen_tsao,2013_wu} primarily focused on achieving the $(k,n)$-threshold,  
but yielded relatively low contrast values.  
Subsequent work by Guo et al.\cite{2013_guo} improved contrast 
by distributing all $n$ shares through repeated $(k, k)$ RGVCS executions, 
while Shyu\cite{2015_shyu} and Yan et al.\cite{2018_yan} introduced cyclic assignment of the first $k$ shares to 
the remaining $n-k$ shares, further enhancing the contrast. 
However, there still remains a gap between 
the contrast of current $(k,n)$ RGVCSs and the theoretical upper bound\cite{14TIFS}.
Additionally, numerous studies\cite{XOR_1,XOR_2} have proposed XOR-based RGVCS to improve contrast via recovery mechanisms. However, XOR recovery requires computational devices, making it less straightforward and efficient than OR-based RGVCS. 

Traditional VCS treat all shares uniformly, resulting in undifferentiated recovery that cannot support hierarchical access control.
To address this limitation, a priority-based VCS model was introduced by Hou \emph{et al.}\cite{priority_hou}, 
where weights are assigned to participants to achieve varying recovery effects, 
with higher weights corresponding to better recovery quality, 
but it suffers from inconsistent share transmittance. 
Subsequent improvements by Yang \emph{et al.}\cite{priority_yang} resulted in a VCS 
with consistent transparency and support for arbitrary privilege levels, 
albeit at the cost of basic matrices design. Building on this, 
Fan \emph{et al.}\cite{priority_fan} proposed an RGVCS with adaptive priority, 
eliminating basic matrices design and simplifying implementation. 
Liu \emph{et al.}\cite{priority_liu} further expanded the functionality, 
enabling support for multiple decryption and lossless recovery.
Nevertheless, in all the aforementioned schemes, each share is assigned a fixed weight, and the recovery quality is entirely determined by the sum of the weights of the selected shares, where higher weights directly yield superior recovery outcomes. 
This fixed-weight mechanism lacks flexibility to dynamically adjust share importance based on actual requirements, and further risks substantial degradation in recovery quality when high-weight shares are lost.

In this paper, we focus on the construction of OR-based $(k,n)$ RGVCS with higher contrast. 
We propose a novel sharing paradigm for RGVCS, 
enabling the construction of $(k,n)$ RGVCS from arbitrary $(k,n')$-threshold schemes where $k \leq n'\leq n$, termed \emph{$n'$-grouped $(k,n)$ RGVCS}. 
Based on this paradigm, we develop a $k$-grouped $(k,n)$ RGVCS and 
successfully achieve a breakthrough in contrast.
The main contributions of this paper are enumerated as follows:
\begin{enumerate}
	\item A novel sharing paradigm for RGVCS is proposed, wherein the quality of recovered images exhibits hierarchical characteristics. Furthermore, we introduce a new contrast calculation formula tailored to this paradigm, which precisely quantifies hierarchical recovery properties. By abandoning fixed-weight allocation in favor of a group-based structure, our scheme achieves hierarchical recovery effects that surpass traditional priority-based VCS in both reconstruction quality and flexibility.
	\item A novel $(k,n)$ RGVCS under the new paradigm is proposed, achieving the currently optimal contrast. Compared to existing schemes, our scheme significantly eliminates the final global shuffling operation and improves the contrast. Furthermore, our scheme demonstrates a hierarchical effect in contrast, 
making it particularly suitable for scenarios where different participants exhibit varying recovery effects.
\end{enumerate}

In the remainder of this paper, Section \ref{section:Preliminaries} discusses the fundamentals and prior research of RGVCS.
A $n'$-grouped $(k,n)$ RGVCS and its contrast analysis are introduced in Sections \ref{section:grouped k_n'_n} and \ref{contrast}, respectively.
Section \ref{section n'=k} provides a detailed and computable formula for 
the contrast of $k$-grouped $(k,n)$ RGVCS.
The experimental results and analysis are presented in Section \ref{section:experiment}. 
Section \ref{section:conclusion} concludes this work.

\section{Preliminaries}{\label{section:Preliminaries}}
In this section, 
we provide the model and algorithms for $(k,n)$ RGVCS, along with 
some relevant definitions and evaluation metrics. 
To facilitate the discussion, 
we first introduce some notations that will be employed throughout this paper.
The symbol $\otimes$ denotes the logical $OR$ operation, which 
corresponds to the superimposition operation in the recovery phase of $(k,n)$ RGVCS. 
The symbol $\oplus$ represents the logical $XOR$ operation.  
Let $T_S$ denote an element randomly selected from the set $S=\{a_1,a_2,\ldots,a_n\}$, 
where $T_S$ takes the value $a_i$ with probability $\frac{1}{n}$ for each $i$ $(1\leq i \leq n)$. 
Let $I_S{(x)}$ denote the indicator function of a set $S$, defined by: 
\begin{equation}
  I_S{(x)} = 
  \begin{cases}
    1, &\text{if}~x \in S\\
    0, &\text{if}~x \notin S
  \end{cases}. \notag
\end{equation}
Let $|S|$ denote the cardinality of $S$.\footnotemark
\footnotetext{
  If $S$ is a multiset, its cardinality is defined as the total count of all element occurrences 
  in the set. For example, given $S = \{0,1,2,1\}$, we have $|S| = 4$. 
}
Additionally, 
let $\pi \in S_n$ be a permutation belonging to the symmetric group on $n$ elements, 
where $S_n$ represents the symmetric group of all possible permutations of $n$ distinct elements. 
For an input bit sequence $(b_1,b_2,\ldots,b_n)$, the permuted output is defined as 
$\pi(b_1,b_2,\ldots,b_n)\stackrel{\triangle}{=} (b_{\pi(1)},b_{\pi(2)},\ldots,b_{\pi(n)})$, 
where $\pi(i)~(1\leq i \leq n)$ denotes the action of the permutation $\pi$ on the index $i$, 
mapping it to a new position in the sequence. 
\vspace{-8pt} 


\subsection{Traditional $(k,n)$ RGVCS Model}
The traditional $(k,n)$ RGVCS consists of two phases: the sharing phase and the recovery phase. 
During the sharing phase, 
the secret image is encrypted into $n$ shadow images of identical size to the original secret image, 
which are then distributed to participants. 
In the recovery phase, 
the secret information can be reconstructed by collecting no fewer than $k$ shadow images 
and performing alignment and superposition operations. 
Taking secret image $S$ as an example, the sharing and recovery phases are described as follows:
\begin{enumerate}
\item The sharing phase: 
when encrypting $S$, a pixel-by-pixel encryption approach is adopted. 
For each pixel $s$, 
first, it is encrypted into $k$ fundamental bits, denoted as $b_1,b_2,\ldots,b_k$, forming the initial bit set $\mathcal{K}$;
then, the remaining $n-k$ bits are generated through specific bit transformation operations
 (including but not limited to random assignment, repeated assignment, or cyclic assignment, etc.);
finally, these $n$ bits are randomly rearranged, 
and sequentially distributed to corresponding positions in each shadow image.
After completing the encryption for every position in $S$, $n$ shadow images are ultimately formed, 
represented as $\mathcal{SC} \stackrel{\triangle}{=} \{SC_1,SC_2,\ldots,SC_n\}$.
\item The recovery phase:
the recovered image $R$ obtained by stacking any $t~(t\geq k)$ shadow images from $\mathcal{SC}$ 
can reveal the secret information. 
\end{enumerate}

\vspace{-8pt} 
\subsection{$(k,n)$ RGVCS Algorithms}
There are numerous encryption algorithms to accomplish the generation of shadow images in the sharing phase, 
with typical examples including 
Chen \& Tsao's scheme\cite{2011_chen_tsao}, Wu \& Sun's scheme\cite{2013_wu}, Yan's scheme\cite{2018_yan}, Shyu's scheme\cite{2015_shyu} and so on. 
Algorithm \ref{algorithm:k_k} and Algorithm \ref{algorithm:k_n_RGVCS} 
describe the encryption steps for a single secret pixel in these schemes. 
For a secret pixel $s$, 
the first step is to execute $(s,k,k)$ RGVCS (Algorithm \ref{algorithm:k_k}) to generate $\mathcal{K}$. 
Subsequently, the generation of the remaining $n-k$ bits varies among different schemes (see step 3 of Algorithm \ref{algorithm:k_n_RGVCS} for details). 
Finally, the $n$ generated bits are randomly rearranged and distributed to each shadow image. 
By applying the same encryption operation to each position of the secret image, 
$\mathcal{SC}$ can be generated. 

During the recovery phase, 
decryption is directly achieved through a simple superimposition method. 

\subsection{Light Transmission}
The random grid was initially proposed by 
Kafri and Keren \cite{1987_kafri_keren} in 1987. They 
described the random grid as a two-dimensional array of pixels,
where each pixel is either completely transparent or entirely opaque, 
determined by a coin-flipping process. 
Transparent pixels allow light to pass through, 
while opaque pixels block light. 
In this paper, we use $0$ to represent transparent pixels 
and $1$ to denote opaque pixels. 
Next, we present the definition of light transmission. 

\begin{definition}
  \cite{2007_shyu},\cite{1987_kafri_keren} (Light Transmission) 
  The light transmission is defined for both an individual pixel and the entire binary image as follows: 

  (1) For a pixel $s$ in binary image $S$, 
  its light transmission $t(s)$ is defined as the probability that the pixel value to be $0$,
  i.e., $t(s)=\Pr(s=0)$. 
  In particular, the light transmission of a transparent pixel is $1$,
  while that of an opaque pixel is $0$.

  (2) For a binary image $S$, the (average) light transmission $T(S)$ is calculated theoretically as follows: 
  \begin{equation}
    \label{eq:transmission}
    T(S) = \frac{1}{h\times w}\sum_{i=1}^{h}\sum_{j=1}^{w}t(S[i,j]),
  \end{equation}
  where $h$ and $w$ denote the height and width of $S$, respectively, 
  $S[i,j]$ represents the pixel value in the position $(i,j)$. 
\end{definition}
In practical experiments, since each pixel is definitively set to either $0$ or $1$, 
eliminating any randomness, the average light transmission of a binary image $S$ is typically calculated as follows: 
\begin{equation}
  T(S) = \frac{N_w}{N_S} = \frac{\sum_{i=1}^{h}\sum_{j=1}^{w}I_{\{0\}}(S[i,j])}{h \times w}, 
\end{equation}
where $N_S$ and $N_w$ denote the total number of pixels and the number of transparent pixels in $S$. 

Naturally, the average light transmission of a random grid is $\frac{1}{2}$, 
resulting in approximately equal numbers of transparent and opaque pixels. 

The following Lemma characterizes the light transmission of set 
$\mathcal{K}$ produced by Algorithm \ref{algorithm:k_k} 
under superimposition with different bit quantities. 
 
\begin{lemma}\cite{2011_chen_tsao}
  \label{lemma:light transmission}
  The light transmission resulting from the superimposition of $q~(q<k)$ and $k$ bits 
  generated from Algorithm \ref{algorithm:k_k} is provided as follows: 
  \begin{equation}
    \begin{aligned}\label{eq:k bits light transmission}
      &t(b'_q[s_{(0)}]) = t(b'_q[s_{(1)}]) = \frac{1}{2^q},\\
      &t(b'_k[s_{(0)}]) = \frac{1}{2^{k-1}}\neq t(b'_k[s_{(1)}]) = 0, 
    \end{aligned}
  \end{equation}
  where $b'_h[s_{(t)}]$ denotes 
  the stacking result of any selected $h$ bits from $\mathcal{K}$ produced by $s=t$ 
  for $h = 1,2,\ldots,k$; $t = 0,1$. 
\end{lemma}

\begin{algorithm}[t]  
  \caption{$(s,k,k)$ RGVCS}
  \label{algorithm:k_k}
  \SetKwData{In}{\textbf{in}}\SetKwData{To}{to}
  \DontPrintSemicolon
  \SetAlgoLined
  \KwIn {a secret bit $s$; a positive integer $k$,where $k\geq 2$}
  \KwOut {$b_1,b_2,\ldots,b_{k}$}
  \Begin{   
        $b_i = T_{\{0,1\}}$ for $1\leq i \leq k-1$\;
        $b_k = s \oplus b_1 \oplus b_2 \oplus \cdots \oplus b_{k-1}$\;    
      \Return{$b_1,b_2,\ldots,b_{k}$}  
  }
\end{algorithm}
\vspace{-8pt} 

\begin{algorithm}[t]  
  \caption{$(s,k,n)$ RGVCS}
  \label{algorithm:k_n_RGVCS}
  \SetKwData{In}{\textbf{in}}\SetKwData{To}{to}
  \DontPrintSemicolon
  \SetAlgoLined
  \KwIn {
    a secret bit $s$; \\
    threshold parameters $(k,n)$, where $2 \leq k \leq n$
  }
  \KwOut {
    $b_{\pi(1)},b_{\pi(2)},\ldots,b_{\pi(n)}$
  }
  \Begin{
      Generate the first $k$ bits:
        \text{\small{$execute$ $(s,k,k)~RGVCS$ $to~ obtain~\mathcal{K}$}}  \;

      Generate the remaining $n-k$ bits:
      \(
        \small
        \left\{
        \begin{array}{l}
          \text{\textbullet}~\text{Chen \& Tsao's scheme:} \\
          \quad b_j = T_{\{0,1\}} \quad (k < j \leq n) \\\\

          \text{\textbullet}~\text{Wu \& Sun's scheme:} \\
          \quad b_j = b_k \quad (k < j \leq n) \\\\

          \text{\textbullet}~\text{Yan's scheme:} \\
          \quad b_{j} = b_{j-k} \quad (k < j \leq n ) \\\\

          \text{\textbullet}~\text{Shyu's scheme:} \\
          {\scriptsize
            \quad b_j = 
            \begin{cases}
              b_{j-k}, &  k < j \leq \lfloor \frac{n}{k} \rfloor k \\
              T_{\mathcal{K}}, & j = \lfloor \frac{n}{k} \rfloor k + 1 \\
              T_{\mathcal{K}\setminus\{b_{i_1},b_{i_2},\ldots,b_{{i_{j-\lfloor \frac{n}{k} \rfloor k-1}}}\}}, & \lfloor \frac{n}{k} \rfloor k+1  < j \leq n\\             
            \end{cases}
          }\\
          \quad \text{where} ~ b_{i_t} = b_{\lfloor \frac{n}{k} \rfloor k + t},~  \scriptsize{1\leq t \leq j-1}\\
        \end{array}
        \right.
      \)\;
      $b_{\pi(1)},b_{\pi(2)},\ldots,b_{\pi(n)}$ = $\pi(b_1,b_2,\ldots,b_n)$\;
    \Return{$b_{\pi(1)},b_{\pi(2)},\ldots,b_{\pi(n)}$}  
    }
\end{algorithm}
\vspace{-10pt} 

\subsection{Contrast}
To quantify the difference between the recovered image and the secret image, 
contrast, represented as $\alpha$, is widely adopted as an evaluation metric in VCS. 
A higher contrast value indicates better visual quality of the recovered image, 
making it more discernible to the human eye.
The specific contrast formula for RGVCS is defined below. 
\begin{definition}
  \cite{2007_shyu}
  \label{definition: contrast of RGVCS}
  The contrast of the recovered image $R$ with respect to the secret image $S$ in RGVCS 
  is calculated as follows: 
  \begin{equation}
    \label{eq:RG contrast}
    \alpha = \frac{T(R(S[0]))-T(R(S[1]))}{1+T(R(S[1]))},
  \end{equation}
  where $R(S[0])$ and $R(S[1])$ represent the regions within $R$ 
  that correspond to the areas with pixel values of $0$ and $1$ in $S$, respectively.
\end{definition}

A valid $(k,n)$ RGVCS is defined based on contrast as follows.

\begin{definition}
  \cite{2007_shyu},\cite{2011_chen_tsao} 
  \label{definition:two conditions}
  A $(k,n)$ RGVCS is called valid if the following two conditions can be satisfied: 
  \begin{itemize}
    \item (Visually recognizable condition)
    The contrast of the recovered image obtained by stacking any $k$ shadow images is greater than $0$, i.e., $\alpha>0$, 
    which implies that the human visual system can recognize the recovered image.
    \item (Security condition)
    The contrast of the recovered image obtained by stacking any $q ~(q<k)$ shadow images is equal to $0$, i.e., $\alpha=0$,
    which indicates that the recovered image does not reveal any information about the secret image.
  \end{itemize}
\end{definition}

\section{$n^{\prime}$-grouped $(k,n)$ RGVCS}\label{section:grouped k_n'_n}
In this section, we propose a novel $(k,n)$ RGVCS under a new sharing paradigm.  
We first introduce the mechanism of the new sharing paradigm, then present an implementation scheme, 
followed by an analysis of its performance, 
including the contrast calculation method and the validity proof. 

\subsection{The new bit-level sharing paradigm for $n'$-grouped $(k,n)$-threshold}
Before the sharing phase begins, 
we first need to select the value of $n^{\prime}$ that satisfies $k\leq n^{\prime} \leq n$. 
In the new sharing paradigm, 
a secret pixel $s~(0~\text{or}~1)$ is encrypted into $n$ share bits $\{b_1,b_2,\ldots,b_n\}$, denoted as $\mathcal{B}_s$, 
which is composed of two parts: the first $n'$ bits and the remaining $n-n'$ bits. 
The first $n'$ bits form group $\mathcal{G}_1$, 
while the remaining $n-n'$ bits constitute group $\mathcal{G}_2,\mathcal{G}_3,\ldots,\mathcal{G}_{\lceil \frac{n}{n'} \rceil}$, 
where $| \mathcal{G}_i | = n'$ for $2 \leq i \leq {\lceil \frac{n}{n'} \rceil - 1}$ and $|\mathcal{G}_{\lceil \frac{n}{n'} \rceil}| = n - n'({\lceil \frac{n}{n'} \rceil - 1})$. 
In other words, 
$\mathcal{B}_s$ is divided into $\lceil \frac{n}{n'} \rceil$ groups: 
$\mathcal{G}_1,\mathcal{G}_2,\ldots,\mathcal{G}_{\lceil \frac{n}{n'} \rceil}$, where 
  \begin{align}
    \mathcal{G}_g =
      \begin{cases}
        \{b_{(g-1)n' + \delta}| 1 \leq \delta \leq n'\}, & 1 \leq g <  \lceil \frac{n}{n'} \rceil \\
        \{b_{(g-1)n' + 1},b_{(g-1)n' + 2},\ldots,b_n \}, & g = \lceil \frac{n}{n'} \rceil
      \end{cases}. \notag
  \end{align}
Particularly, we define a group of length $n'$ as a complete group, 
while a group with length less than $n'$ is considered an incomplete group.

We begin by generating the first $n'$ bits through $(s,k,n^{\prime})$ RGVCS, 
producing $\mathcal{G}_1 = \{b_{\pi(1)},b_{\pi(2)},\ldots,b_{\pi(n')}\}$. 
For the generation of the remaining $n-n'$ bits, i.e., $\mathcal{G}_g ~(2\leq g \leq \lceil \frac{n}{n'} \rceil)$, 
generate in the following manner: 
for $b_{(g-1)n'+1}$, randomly selected from $\mathcal{G}_1$, i.e., $b_{(g-1)n'+1} = T_{\mathcal{G}_1}$; 
for other bits in $\mathcal{G}_g$, also selected from $\mathcal{G}_1$, but exclude bits selected in previous positions of the same group, 
i.e., $b_{(g-1)n' + \delta} = T_{\mathcal{G}_1 \setminus \{b_{(g-1)n'+1},b_{(g-1)n'+2},\ldots,b_{(g-1)n'+\delta-1}\}}$ 
for $2 \leq \delta \leq \min (n',n-(g-1)n')$. 

Hence, 
each complete group constitutes a random permutation of $\{b_1,b_2,\ldots,b_{n'}\}$, 
while when $n$ is not divisible by $n'$, $\mathcal{G}_{\lceil \frac{n}{n'} \rceil}$ becomes a permuted subset of $\{b_1,b_2,\ldots,b_{n'}\}$. 

To better understand this new sharing paradigm, 
we provide the following example. 
\begin{example}
  Consider $k = 2, n' = 3$, and $n = 8$. 
  The $(s,2,3)$ RGVCS first generates the initial group 
  $\mathcal{G}_1=\{b_{\pi_0(1)},b_{\pi_0(2)},b_{\pi_0(3)}\}=\{b_2,b_1,b_3\}$, 
  where $\pi_0 =
  \left(
    \begin{array}{l}
      1~2~3\\
      2~1~3
    \end{array}
  \right)$.
  The remaining 5 bits are generated as follows:
  
  \noindent\textbf{$\mathcal{G}_2=\{b_4,b_5,b_6\}$}:
  \begin{itemize}
      \item $b_4$: Randomly select from $\{b_2,b_1,b_3\}$ (e.g., $b_1$);
      \item $b_5$: Select from remaining $\{b_2,b_3\}$ (e.g., $b_3$);
      \item $b_6$: Last remaining bit $b_2$.
  \end{itemize}
  
  \noindent\textbf{$\mathcal{G}_3=\{b_7,b_8\}$}:
  \begin{itemize}
      \item $b_7$: Randomly select from $\{b_2,b_1,b_3\}$ (e.g., $b_2$);
      \item $b_8$: Select from remaining $\{b_1,b_3\}$ (e.g., $b_1$).
  \end{itemize}
  
  Thus, $\mathcal{G}_2 = \{b_1,b_3,b_2\}$ and $\mathcal{G}_3 = \{b_2,b_1\}$, 
  $\mathcal{B}_s = \{b_2,b_1,b_3,b_1,b_3,b_2,b_2,b_1\}$. 
\end{example}
By applying the described sharing paradigm to encrypt every pixel in the secret image, 
we can construct a fundamentally new $(k,n)$ RGVCS, 
which will be presented in the following subsection. 

\vspace{-10pt} 
\subsection{The sharing and recovery phases} 
Building upon the new bit-level sharing paradigm, 
we have designed a novel $(k,n)$ RGVCS named $n'$-grouped $(k,n)$ RGVCS. 
The detailed sharing steps are presented in Algorithm \ref{algorithm:n' grouped k_n}, 
where $\mathcal{P}[i,j]$ is used to record the bit indices that have already been selected within the same group at position $(i,j)$ in shadow images.  
Note that $\mathcal{P}[i,j]$ must be reset to empty at the beginning of each new group.
A simple superposition is employed for the recovery phase. 

After executing Algorithm \ref{algorithm:n' grouped k_n}, 
$\mathcal{SC} = \{SC_1,SC_2,\ldots,SC_n\}$ can be obtained. 
Since the new sharing paradigm is applied to each pixel of the secret image, 
we can divide $\mathcal{SC}$ into $\lceil \frac{n}{n'} \rceil$ groups: 
$G_1,G_2,\ldots,G_{\lceil \frac{n}{n'} \rceil}$, where $G_g =$
  \begin{align}
      \begin{cases}
        \{SC_{(g-1)n' + \delta}| 1 \leq \delta \leq n'\},  &1 \leq  g <  \lceil \frac{n}{n'} \rceil \\
        \{SC_{(g-1)n' + 1},SC_{(g-1)n' + 2},\ldots,SC_n \},   &g = \lceil \frac{n}{n'} \rceil
      \end{cases}. \notag
  \end{align}


\begin{figure*}[t]
  \centering
  \captionsetup[subfloat]{font=normalsize} 
  \subfloat[The sharing phase]{
      \includegraphics[width=\textwidth]{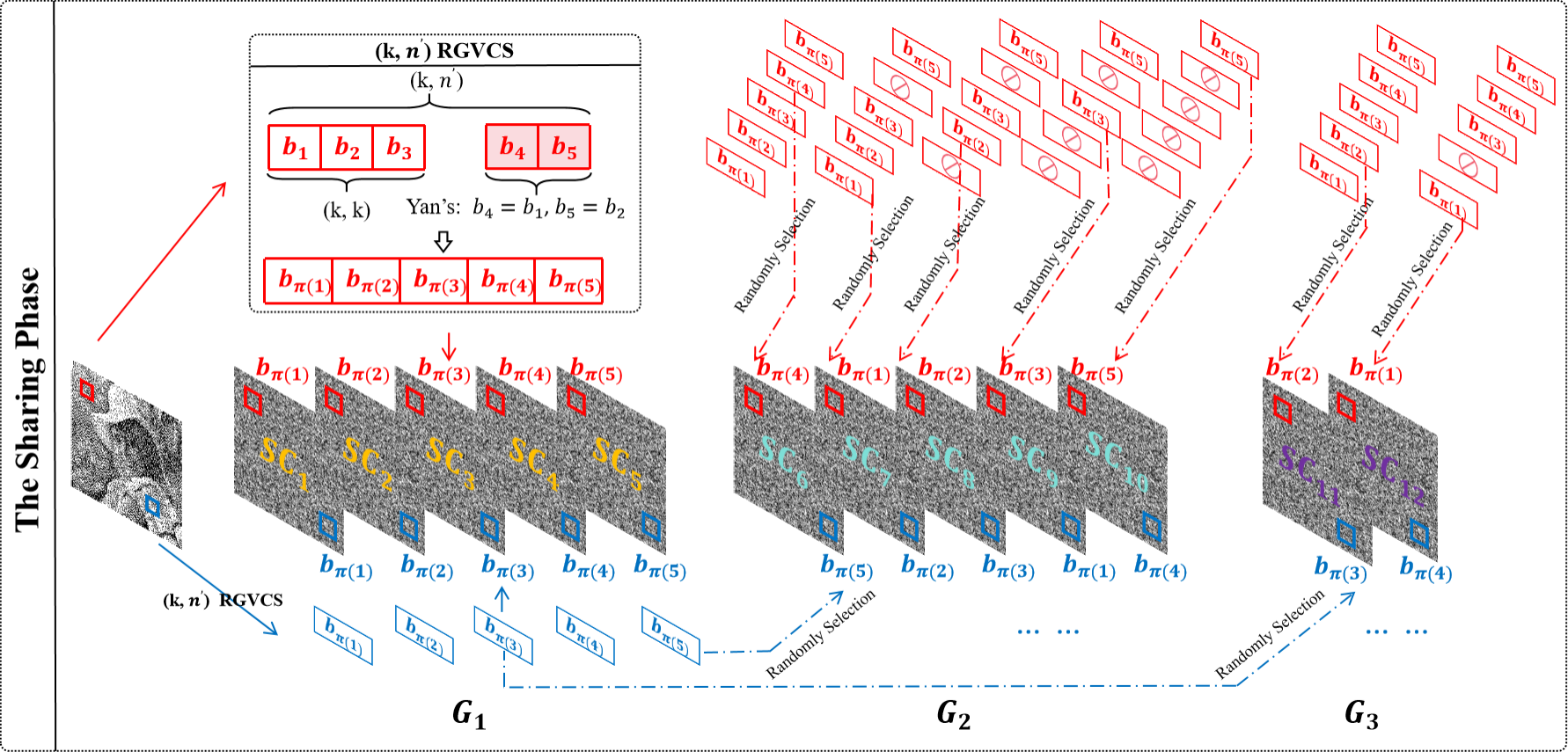}}
      \\
  \subfloat[The recovery phase ]{ 
      \includegraphics[width=\textwidth]{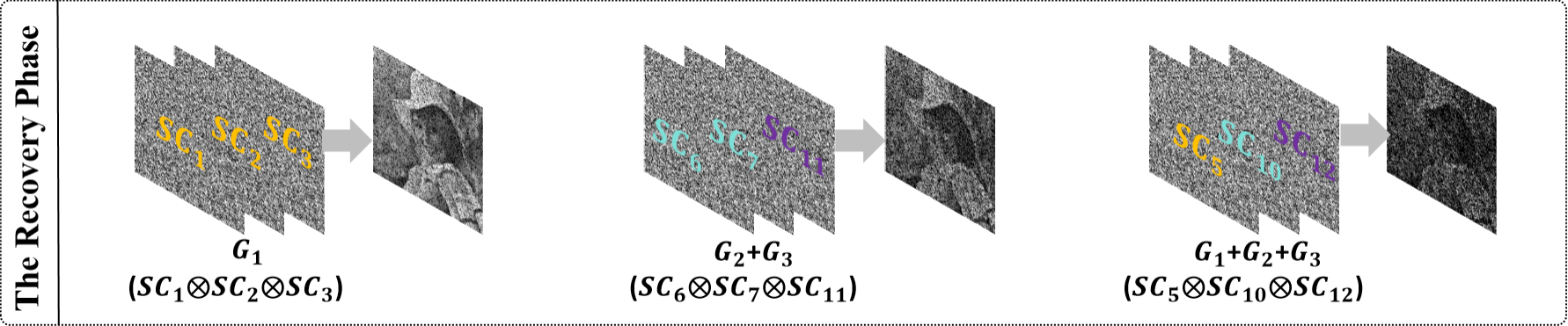}}
  \caption{The sharing phase and recovery phase of $5$-grouped $(3,12)$ RGVCS}
  \label{fig:sharing_recovery}
  \vspace{-10pt} 
\end{figure*}
Figure \ref{fig:sharing_recovery} illustrates the sharing and recovery phases of $5$-grouped $(3,12)$ RGVCS. 
During the sharing phase, the shadow images are divided into $\lceil \frac{12}{5} \rceil~(=3)$ groups: 
$G_1=\{SC_1,SC_2,\ldots,SC_5\}$, $G_2=\{SC_6,SC_7,\ldots,SC_{10}\}$ and $G_3=\{SC_{11},SC_{12}\}$.
Taking the red secret pixel at the upper-left corner $s_{red}$ as an example, 
it is first encrypted via Yan's $(s_{red},3,5)$ RGVCS to generate $\{b_{\pi(1)},b_{\pi(2)},\ldots,b_{\pi(5)}\}$ as $\mathcal{G}_1$, 
which are then distributed to the corresponding positions in $G_1$. 
For the corresponding positions in $G_2$ and $G_3$, 
they are randomly selected from the bit pool composed of $\mathcal{G}_1$, 
where the pool dynamically excludes bits previously selected within the same group.
For other pixel positions in the secret image (e.g., the blue pixel at the bottom-right corner), 
the encryption method remains consistent with $s_{red}$, 
but the bit selection from $\mathcal{G}_1$ is performed completely independently. 

During the recovery phase, as shown by (b) in Figure \ref{fig:sharing_recovery}, 
intra-group shadow superposition achieves optimal recovery quality, 
while inter-group superposition (across 2 groups) demonstrates intermediate effectiveness. 
Notably, hybrid combinations involving 3 groups yield the poorest reconstruction performance. 
Detailed contrast analysis is presented in the following subsection.

The most significant difference compared to traditional $(k,n)$ RGVCS model lies in 
whether the final random permutation operation (step 4 of Algorithm \ref{algorithm:k_n_RGVCS}) is required. 
In traditional model, taking Yan's RGVCS as an example, 
the encryption result for the secret pixel $s$ is 
$\{b_1,b_2,b_3,b_1,b_2\}$ (before random scrambling) under $(3,5)$-threshold. 
After performing identical encryption operations on each pixel of the secret image independently, 
certain shadow images may become completely identical (e.g., $SC_1$ and $SC_4$), 
resulting in failure to recover the secret image. 
Therefore, the traditional model requires a final random scrambling operation in bit-level encryption. 
In contrast, our scheme incorporates randomness during the encryption of each secret pixel, 
ensuring distinct encryption outcomes for different secret pixels. 
For instance, in 3-grouped $(3,5)$ RGVCS, 
two distinct secret pixels may yield $\{b_{\pi(4)},b_{\pi(1)},b_{\pi(2)},b_{\pi(3)},b_{\pi(5)}\}$ and 
$\{b_{\pi(5)},b_{\pi(2)},b_{\pi(3)},b_{\pi(1)},b_{\pi(4)}\}$ as their respective encryption results. 
Consequently, our approach eliminates the need for a final permutation operation at the bit-level encryption. 

Furthermore, the value of $n^{\prime}$ in our scheme is dynamically adjustable. 
When $n^{\prime}=n$, the scheme degenerates into a traditional $(k,n)$ RGVCS, 
and when $k\leq n^{\prime}< n$, 
the contrast of our scheme will exhibit varying degrees of change. 
Section \ref{section n'=k} explores the particular scenario when $n^{\prime}$ is taken to be $k$.

\section{Contrast analysis for $n'$-grouped $(k,n)$ RGVCS} \label{contrast}
This section begins by comparing the proposed scheme with traditional $(k,n)$ RGVCSs, 
followed by introducing definitions and notations to establish the theoretical framework. 
We then present the contrast calculation method of the proposed scheme 
and conclude with a validity analysis of the scheme. 

\vspace{-10pt} 
\subsection{Comparison with traditional $(k,n)$ RGVCSs}

The traditional $(k,n)$ RGVCS and our proposed $n'$-grouped $(k,n)$ RGVCS share the following common characteristic: 
each secret pixel is encrypted in the same method, 
which leads Eq.(\ref{eq:transmission}) to transform into

\begin{small}
  \begin{equation}
    T(S) = \frac{1}{h\times w}\sum_{i=1}^{h}\sum_{j=1}^{w}t(S[i,j]) = t(S[i,j]), \notag
  \end{equation}
\end{small}
where $i\in\{1, 2, \cdots, h\}$ and $j\in\{1, 2, \cdots, w\}$ are any given positions
Consequently, the contrast expression in Eq.(\ref{eq:RG contrast}) can be formulated as: 

\begin{small}
  \begin{equation}
    \label{eq:contrast_single_light_transmission}
    \alpha_{RG} = \frac{T(R(S[0]))-T(R(S[1]))}{1+T(R(S[1]))} = \frac{t(R_{s=0})-t(R_{s=1})}{1+t(R_{s=1})}, 
  \end{equation}
\end{small}
where $R_{s=0}$ and $R_{s=1}$ respectively denote the corresponding pixels in the recovered image $R$ for 
the points where the secret pixel values are 0 and 1. 
That is to say, the contrast calculation for both the traditional and our proposed schemes 
can be based on the single-point light transmission. 

For the differences between the two schemes, 
we demonstrate them through the experimental data statistics for the $5$-grouped $(3,12)$ RGVCS in Table \ref{tab:3_layers}.
When selecting 3 shadow images for recovery, 
there are $\binom{12}{3}$ $(=220)$ different combinations. 
By visualizing the contrast corresponding to all combinations, as illustrated in Figure \ref{fig:3_layers}, 
where the horizontal axis records the combinations of all shadow image indices, 
for example, $"2-6-10"$ represents the combination of the second, sixth, and tenth shadow images.  
We found that they fall into 3 layers. 
The number of combinations in each layer, the contrast range, and the mean values are shown in Table \ref{tab:3_layers}. 
In particular, we analyzed the group distribution of shadow image combinations in each layer: 
the first-layer combinations originate from a single group, categorized as $3=3$, 
the second-layer combinations involve two shadow images from one group and one from another, categorized as $3=2+1$, 
and the third-layer combinations are fully cross-group, with 3 shadow images drawn from distinct groups, categorized as $3=1+1+1$. 
Conversely, 
the traditional RGVCS yields a uniform contrast for all possible combinations, 
exhibiting no layered effect. 

\begin{algorithm}[t]
  \caption{$n^{\prime}$-grouped $(k,n)$ RGVCS}
  \label{algorithm:n' grouped k_n}
  \SetKwData{In}{\textbf{in}}\SetKwData{To}{to}
  \DontPrintSemicolon
  \SetAlgoLined
  \KwIn {an $h\times w$ binary secret image $S$;\\ threhold parameters $k,n',n$, where $2 \leq k\leq n' \leq n$}
  \KwOut {$SC_1,SC_2,\ldots,SC_{n}$}
  \Begin{
      \For{$(i,j)$ where $1\leq i \leq h ,1\leq j \leq w$}{
          Execute $(S[i,j],k,n^{\prime})$ RGVCS to obtain $b_{\pi(1)},b_{\pi(2)},\ldots,b_{\pi(n')}$ and distribute to 
          $SC_1[i,j],SC_2[i,j],\ldots,SC_{n^{\prime}}[i,j]$\;
          \For{$t \gets n^{\prime}+1$ \To $n$}{
         
            \If{$t\!\!\mod n'=1$}{
                $\mathcal{P}[i,j]=\emptyset$
            }
            $p = T_{\{1,2,\ldots,n'\}\setminus \mathcal{P}[i,j]}$\;
            $SC_t[i,j] = b_p$\;
            $\mathcal{P}[i,j].\mathrm{append }(p)\footnotemark$\;            
          } 
      }
      \Return{$SC_1,SC_2,\ldots,SC_{n}$}  
  }
\end{algorithm}
\footnotetext{
              $S.\mathrm{append}(x)$ represents the operation of adding $x$ to set $S$, 
              for example, if $S = \{1,4\}$ and $x=2$, after $S.\mathrm{append}(x)$, 
              $S$ becomes $\{1,2,4\}$.
              All $\mathrm{append}$ operations in this paper follow the same semantics.}

\begin{table}[!t] 
    \centering
    \caption{Experimental data for $5$-grouped $(3,12)$ RGVCS}
    \label{tab:3_layers}
    \begin{adjustbox}{width=0.6\columnwidth}
      {\fontsize{15}{12}\selectfont
      \begin{tabular}{clccclccc}
      \toprule
      Interval             &  & Count                & Contrast Range       & Mean                 &  & 3=3                & 3=2+1            & 3=1+1+1        \\ \midrule
      \multicolumn{1}{l}{} &  & \multicolumn{1}{l}{} & \multicolumn{1}{l}{} & \multicolumn{1}{l}{} &  & \multicolumn{1}{l}{} & \multicolumn{1}{l}{} & \multicolumn{1}{l}{} \\
      1                    &  & 20                   & 0.085502-0.088227    & 0.086508             &  & 20                   & 0                    & 0                    \\
      \multicolumn{1}{l}{} &  & \multicolumn{1}{l}{} & \multicolumn{1}{l}{} & \multicolumn{1}{l}{} &  & \multicolumn{1}{l}{} & \multicolumn{1}{l}{} & \multicolumn{1}{l}{} \\
      2                    &  & 150                  & 0.047080-0.051495    & 0.049068             &  & 0                    & 150                  & 0                    \\
      \multicolumn{1}{l}{} &  & \multicolumn{1}{l}{} & \multicolumn{1}{l}{} & \multicolumn{1}{l}{} &  & \multicolumn{1}{l}{} & \multicolumn{1}{l}{} & \multicolumn{1}{l}{} \\
      3                    &  & 50                   & 0.036622-0.039662    & 0.038210             &  & 0                    & 0                    & 50                   \\ \bottomrule
      \end{tabular}
      }
    \end{adjustbox}
    \vspace{-10pt} 
\end{table}

\begin{figure}[t]
  \centering
  \includegraphics[width=0.8\textwidth,keepaspectratio]{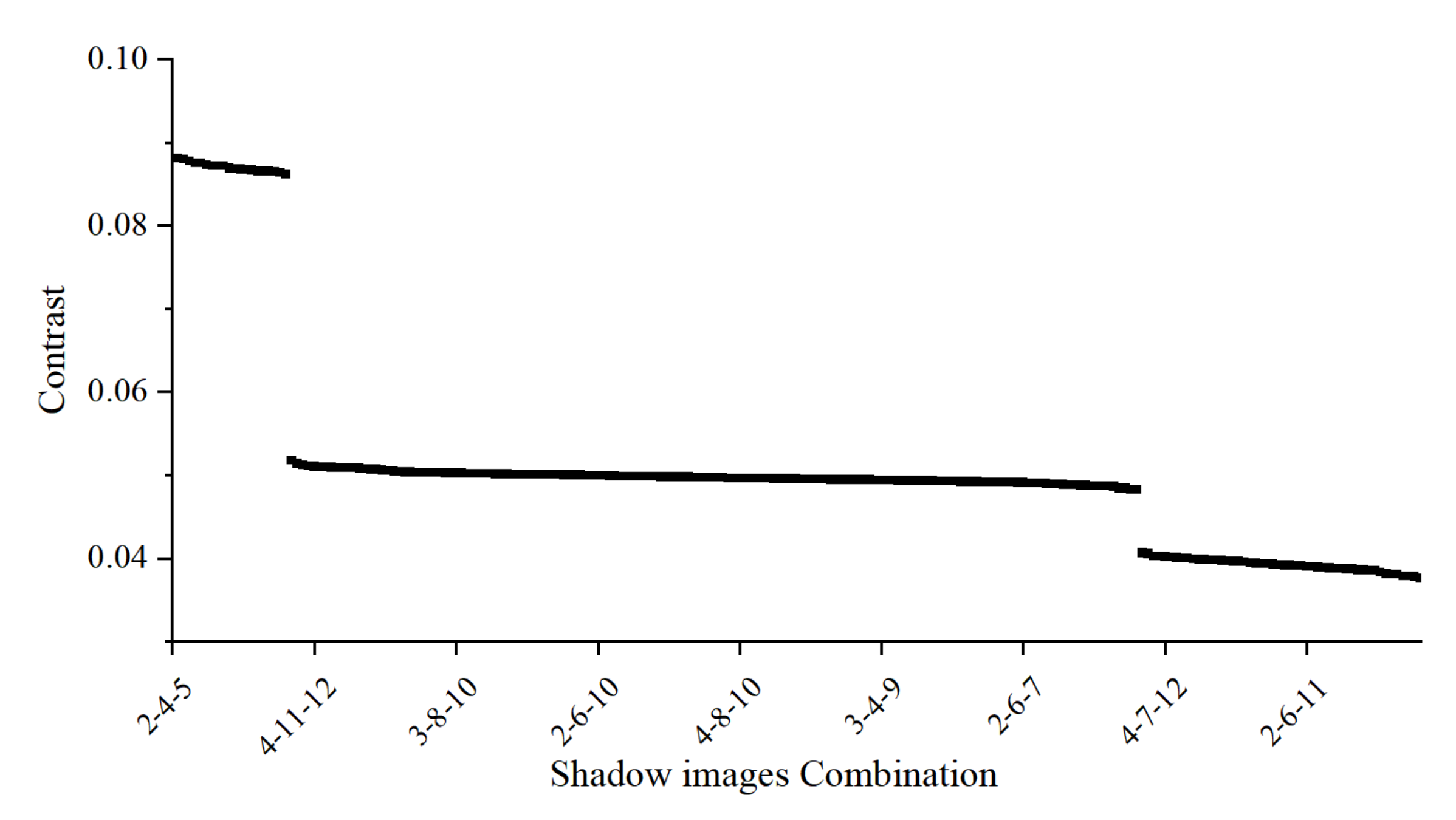}
  \caption{The contrast of all participant combinations in $5$-grouped $(3,12)$ RGVCS}
  \label{fig:3_layers}
  \vspace{-10pt} 
\end{figure}

\subsection{Definitions and Notations}
This subsection formally presents essential definitions and notations 
to establish the mathematical foundations for the subsequent theoretical derivations. 
It should be noted that the notations $k$, $n$, and $n'$ are fixed throughout the following sections, 
with definitions: 
$k$ and $n$ are threshold parameters, 
and $n'$ denotes the length of a share bit group $\mathcal{G}_i~(1\leq i \leq \lceil \frac{n}{n'} \rceil)$. 

\begin{definition}(Integer partition)\label{definition:integer partition}
  Let $t$ be a positive integer, 
  $\vec{\lambda} = (\lambda_1,\lambda_2,\ldots,\lambda_m)$ is called an integer partition of $t$ 
  if each $\lambda_i$ $(1\leq i \leq m)$ is a non-negative integer, 
  and satisfying $\sum_{i=1}^{m}\lambda_i = t$, 
  where $m$ is any positive integer. 
\end{definition}

In $n'$-grouped $(k,n)$ RGVCS, 
the numerical value of non-negative integers should be considered, 
which motivates the following definition of valid partition.

\begin{definition}(Valid partition)\label{definition:valid partition}
  Let $\vec{\lambda} = (\lambda_1,\lambda_2,\ldots,\lambda_{\lceil \frac{n}{n'} \rceil})$ be a partition of $t$, 
  $\vec{\lambda}$ is called valid if satisfies: 
  \begin{align}\label{eq:valid partition}
    \begin{cases}
    \max_{1\leq i \leq \lceil \frac{n}{n'} \rceil} \lambda_i \leq n'\\
    \lambda_{\lceil \frac{n}{n'} \rceil} \leq |\mathcal{G}_{\lceil \frac{n}{n'} \rceil}| 
  \end{cases}.
  \end{align}
\end{definition}

\begin{remark}
  The two conditions in Eq.(\ref{eq:valid partition}) are imposed because: 
  (i) at most $n'$ bits can be selected from each share bit group; 
  (ii) the group $\mathcal{G}_{\lceil \frac{n}{n'} \rceil}$ may be incomplete, 
  from which a maximum of $|\mathcal{G}_{\lceil \frac{n}{n'} \rceil}|$ bits can be selected. 
\end{remark}

To better understand Definition \ref{definition:valid partition}, an example is presented. 

\begin{example}\label{example:valid partitions}
  Consider the case where $n'=3,n=7$, and $t=3$. 
  $\mathcal{B}_s$ is consists of 3 groups: $\mathcal{G}_1,\mathcal{G}_2,\mathcal{G}_3$, 
  where $|\mathcal{G}_1| = |\mathcal{G}_2| = 3$, and $|\mathcal{G}_3| = 1$. 
  The valid partitions of $t$ are: $(3,0,0)$, $(0,3,0)$, $(0,2,1)$, $(2,0,1)$, $(2,1,0)$, $(1,2,0)$ and $(1,1,1)$. 
  While partitions $(0,0,3)$, $(0,1,2)$, and $(1,0,2)$ fail to satisfy the second condition in Eq.(\ref{eq:valid partition}), 
  they cannot be considered valid partitions. 
\end{example}

Then, we define the index multiset of the share bits selected based on a given valid partition. 

\begin{definition} (Index multiset) \label{definition: index set}
  Let $I \stackrel{\triangle}{=} \{g_1,g_2,\ldots,g_{\lceil \frac{n}{n'} \rceil}\}$ 
  be the index multiset corresponding to $\mathcal{B}_s$, 
  where each $g_i~(1\leq i \leq {\lceil \frac{n}{n'} \rceil})$ is a multiset containing 
  the indices of all bits belonging to $\mathcal{G}_i$. 
\end{definition}

\begin{example}\label{example:index multiset}
  Consider $\mathcal{B}_s$ concludes the following 3 share bit groups: 
  $\mathcal{G}_1 = \{b_1,b_2,b_3,b_1,b_2\}$, $\mathcal{G}_2 = \{b_3,b_1,b_2,b_2,b_1\}$, 
  and $\mathcal{G}_3 = \{b_3,b_1\}$. 
  Then, $g_1 = \{1,2,3,1,2\}$, $g_2 = \{3,1,2,2,1\}$ and $g_3 = \{3,1\}$. 
  Thus its corresponding index multiset is $I \stackrel{\triangle}{=} \{\{1,2,3,1,2\},\{1,2,3,1,2\},\{1,3\}\}$. 
\end{example}

Next, we present the definition of another multiset, 
which consists of $t$ bit indices selected from the multiset $I$ 
according to a given valid partition of $t$. 

\begin{definition} (Selected index multiset)
  Let $C(\vec{\lambda})$ denote the selected index multiset, 
  its selection is performed as follows: 
  \begin{equation}
    C(\vec{\lambda})  \triangleq \bigcup_{j = 1}^{\lceil \frac{n}{n'} \rceil} S[g_j, \lambda_j], \notag
  \end{equation}
  where $S[g_j, \lambda_j]$ represents the multiset formed by selecting 
  $\lambda_j$ $(1\leq j \leq \lceil \frac{n}{n'} \rceil)$ elements from $g_j$ simultaneously, 
  and 
  let $A_j$ denote the set corresponding to $S[g_j,\lambda_j]$, i.e.,
  \begin{equation}
    A_j  \triangleq \{x~ |~ x \in S[g_j,\lambda_j]\}. \notag
  \end{equation}
  Particularly, 
  let $\#C(\vec{\lambda})$ denote the number of distinct elements in $C(\vec{\lambda})$, i.e.,
  $\#C(\vec{\lambda}) =\big |\bigcup_{j = 1}^{\lceil \frac{n}{n'} \rceil} A_j \big | $.
\end{definition}

\begin{example} \label{example: selected index multiset}
  Given a valid partition $\vec{\lambda} = (2,2,1)$ 
  and the same $g_1,g_2,$ and $g_3$ as in Example \ref{example:index multiset}, 
  then the selected index multiset $C(\vec{\lambda})$ may fall into the following case: 
  $C(\vec{\lambda}) = \{1,1\}\cup \{1,2\} \cup \{3\} = \{1,1,1,2,3\}$. 
  Then, $A_1 =\{1\}$, $A_2 =\{1,2\}$, $A_3=\{3\}$, and 
  $\#C(\vec{\lambda}) = 3$. 
\end{example}
\begin{remark}
  If the selected bit indices are the same, 
  it means that the same bits have been selected, and their
  superimposition result is their bit value. 
  Consequently, the light transmission after bit superimposition is determined exclusively by the size of $\#C(\vec{\lambda})$. 
\end{remark}

\begin{definition}
  Let $t_i~(1\leq i \leq \lceil \frac{n}{n'} \rceil)$ be the number of new elements appearing in $A_i$, 
  defined by 
  \begin{equation} \label{eq11}
    t_{i} \triangleq \Big|\bigcup_{j=1}^{i}A_j\Big| - \Big|\bigcup_{j=1}^{i-1}A_j\Big|, 
  \end{equation}
  for every $i\in\{1,2,\cdots,\lceil \frac{n}{n'} \rceil\}$. In particular, $t_{1}= \left|A_1\right| $.
  Let 
  \begin{equation} 
    L_h \triangleq \Big\{(t_1,t_2,\ldots,t_{\lceil \frac{n}{n'} \rceil}) \Big| \sum_{i =1}^{\lceil \frac{n}{n'} \rceil} t_i = h\Big\} \notag
  \end{equation}
  denote the set of all possible $t_i$ sequences satisfying $\#C(\vec{\lambda})=\sum_{i =1}^{\lceil \frac{n}{n'} \rceil} t_i = h$, where $1\leq h \leq k$. 
\end{definition}


\subsection{Contrast calculation method}
The contrast calculation method for $n'$-grouped $(k,n)$ RGVCS is discussed in this subsection. 
We first analyze the different cases of bit selection from complete or incomplete bit groups. 

\begin{lemma}\label{lemma:complete_incomplete}
  Let $\mathcal{G} \triangleq\{b_1,b_2,\ldots,b_{n'}\}$ be a complete bit group, and let $\mathcal{G}' \triangleq \{b'_1,b'_2,\ldots,b'_l\}\subseteq \mathcal{G}$ be an incomplete bit group consisting of $l$ bits randomly selected from $\mathcal{G}$, where $1 \leq l < n'$. Let $m$ be a positive integer less than or equal to $l$.
  Let $A_1$ denote the event of randomly selecting $m$ elements from $\mathcal{G}$ that are exactly $b_{i_1},b_{i_2},\ldots,b_{i_m}$, and let $A_2$ denote the event of randomly selecting $m$ elements from $\mathcal{G}'$ that are $b_{i_1},b_{i_2},\ldots,b_{i_m}$. Then, $\Pr(A_1)=\Pr(A_2)$.
\end{lemma}
\begin{proof}
  First, it is easy to obtain $\Pr(A_1)= \frac{1}{\binom{n'}{m}}$. 
  Second, for event $A_2$ to occur, $\mathcal{G}'$ must contain $b_{i_1},b_{i_2},\ldots,b_{i_m}$, 
  and then $b_{i_1},b_{i_2},\ldots,b_{i_m}$ are selected from $\mathcal{G}'$. Therefore, we obtain
 \begin{equation*}
 	\begin{aligned}
 		\Pr(A_2)  &=\frac{\binom{n'-m}{l-m}}{\binom{n'}{l}}\cdotp\frac{1}{\binom{l}{m}} = \dfrac{\frac{(n'-m)!}{(n'-l)!(l-m)!}}{\frac{n'!}{l!(n'-l)!}\cdotp\frac{l!}{m!(l-m)!}}  \\
 		    &= \frac{m!(n'-m)!}{n'!} =\frac{1}{\binom{n'}{m}} =\Pr(A_1).
 	\end{aligned}
 \end{equation*}
\end{proof}
Lemma~\ref{lemma:complete_incomplete} tells us that the probability of selecting a given m elements from the complete group and the incomplete group is the same. Furthermore, we obtain the following lemma.
\begin{lemma}\label{lemma:permuted partition}
  For a given valid partition $\vec{\lambda} = (\lambda_1,\lambda_2,\ldots,\lambda_{\lceil \frac{n}{n'} \rceil})$ 
  and a permutation $\varphi$, 
  let $\varphi(\vec{\lambda})\triangleq (\lambda_{\varphi(1)},\lambda_{\varphi(2)},\ldots,\lambda_{\varphi(\lceil \frac{n}{n'} \rceil)})$ 
  be a valid partition. 
  For an integer $h~(1\leq h \leq k)$, then
  \begin{equation*}
    \Pr(\#C(\vec{\lambda}) = h) = \Pr(\#C(\varphi(\vec{\lambda})) = h).  
  \end{equation*} 
\end{lemma}
\begin{proof}
	Since Lemma~\ref{lemma:complete_incomplete} holds and the selection from each bit group is completely independent, we can simplify the proof of Lemma 3 by considering only the case where each group is a complete bit group; that is, the case where $n\equiv 0\ (\text{mod\ } n')$.
	 
  We first introduce several definitions and state three claims. 
  Let $\hat{g}_i$ and $\hat{\lambda}_i~ (1 \leq i \leq \lceil \frac{n}{n'} \rceil)$ denote 
  the $i$-th selected index multiset after permutation and the corresponding number of selected elements, respectively. We define $\hat{g}_i \triangleq g_{\varphi^{-1}(i)}$ to determine the order after permutation.
  Similarly, we can define
  \begin{equation*}
  	\begin{aligned}
    C(\varphi(\vec{\lambda})) &\triangleq \bigcup_{j = 1}^{\lceil \frac{n}{n'} \rceil} S[\hat{g}_j, \hat{\lambda}_j], \footnotemark \\
     \hat{A}_j & \triangleq \{x~ |~ x \in S[\hat{g}_j,\hat{\lambda}_j]\}, \\
    \hat{L}_h &\triangleq \{(\hat{t}_{1}, \hat{t}_{2}, \ldots, \hat{t}_{\lceil \frac{n}{n'} \rceil}) | \sum_{i=1}^{\lceil \frac{n}{n'} \rceil}\hat{t}_{i}= h \}, 
    \end{aligned} 
    \footnotetext{Due to $\bigcup_{j = 1}^{\lceil \frac{n}{n'} \rceil} S[\hat{g}_j, \hat{\lambda}_j]= \bigcup_{j = 1}^{\lceil \frac{n}{n'} \rceil} S[g_j, \lambda_{\varphi(j)}]$, we define the order after the permutation, which does not affect the correctness of the proof.}
  \end{equation*} 
  where $\hat{t}_{i} \triangleq |\bigcup^{i}_{j=1}\hat{A}_{j}| - |\bigcup^{i-1}_{j=1}\hat{A}_{j}|$. We claim as follows.
  \begin{claim}\label{claim:1}
  	When $n\equiv 0\ (\text{mod\ } n')$, then
    $\hat{g}_i = g_{i}, ~\hat{\lambda}_i=\lambda_i$, for $1\leq i \leq \lceil \frac{n}{n'} \rceil$.
  \end{claim}
  \begin{claim}\label{claim:2}
  	When $n\equiv 0\ (\text{mod\ } n')$, then $L_h \!=\! \hat{L}_h$, for $1\leq h \leq k$.
  \end{claim}
  \begin{claim}\label{claim:3}
  	Let $X_i$ and $\hat{X}_i$ be random variables denoting the number of new elements appearing when selecting the $i$-th index set before and after permutation, respectively.
    For any given sequence $\{a_i\}_{i=1}^{\lceil \frac{n}{n'} \rceil}$, then
     \begin{equation*}
    	\begin{aligned}
   & \Pr(X_1 = a_1, X_2 = a_2, \ldots, X_{\lceil \frac{n}{n'} \rceil} = a_{\lceil \frac{n}{n'} \rceil}) \\
     = & \Pr(\hat{X}_1 = a_1, \hat{X}_2 = a_2, \ldots, \hat{X}_{\lceil \frac{n}{n'} \rceil} = a_{\lceil \frac{n}{n'} \rceil}),
     \end{aligned}
      \end{equation*} 
     where $1\leq a_i \leq k$ for every $1\leq i \leq \lceil \frac{n}{n'} \rceil$.
  \end{claim}
  \textbf{The proof of Claim \ref{claim:1}:} 
  Since all groups are complete bit groups, $g_i=g_j$ for every $1\leq i\leq j\leq\lceil \frac{n}{n'} \rceil$.
  Thus, $\hat{g}_i = g_{\varphi^{-1}(i)}=g_i$. In addition, based on the defined order $\hat{g}_i = g_{\varphi^{-1}(i)}$, it follows that $\hat{\lambda}_i = \lambda_{\varphi(\varphi^{-1}(i))} = \lambda_i$.

  \textbf{The proof of Claim \ref{claim:2}:} For any $\vec{t} = (t_1,t_2,\ldots,t_{\lceil \frac{n}{n'} \rceil}) \in L_h$, there exist $A_1,A_2,\ldots,A_{\lceil \frac{n}{n'}\rceil}$ satisfying Eq.~\eqref{eq11}.
   By Claim \ref{claim:1}, we know that $S[\hat{g}_i,\hat{\lambda}_i]=S[g_i,\lambda_i]$ can be obtained by selection. Therefore, due to the definition of $A_i$, we have
   \begin{equation*}
   	 A_i = \{x~ |~ x \in S[g_i,\lambda_i]\} = \{x~ |~ x \in S[\hat{g}_i,\hat{\lambda}_i]\}= \hat{A}_i   
   \end{equation*}
  for any $1\leq i\leq \lceil \frac{n}{n'}\rceil$. Furthermore, we obtain
  \[
   \hat{t}_{i} =\Big |\bigcup^{i}_{j=1}\hat{A}_{j}\Big| - |\bigcup^{i-1}_{j=1}\hat{A}_{j}|=\Big|\bigcup_{j=1}^{i}A_j\Big| - \Big|\bigcup_{j=1}^{i-1}A_j\Big|=t_i
  \]
  for any $1\leq i\leq \lceil \frac{n}{n'}\rceil$. Thus, $\vec{t} = (t_1,t_2,\ldots,t_{\lceil \frac{n}{n'} \rceil})=(\hat{t}_1,\hat{t}_2,\ldots,\hat{t}_{\lceil \frac{n}{n'} \rceil}) \in \hat{L}_h$, and hence $L_h \subseteq \hat{L}_h$.
  
  The converse holds similarly, yielding $\hat{L}_h \subseteq L_h$. Therefore, $L_h = \hat{L}_h$.

 \textbf{ The proof of Claim \ref{claim:3}:} Due to the definition of $X_i$ and $\hat{X}_i$ and Claim \ref{claim:1}, we obtain
  \begin{equation*}
  	\begin{aligned}
  		&\Pr\big(X_i = a_i \big|X_{i-1} = a_{i-1},X_{i-2} = a_{i-2}, \ldots, X_1 = a_1\big) \\
  	   =&\Pr\big(\hat{X}_i=a_i \big| \hat{X}_{i-1} = a_{i-1}, \hat{X}_{i-2} = a_{i-2},\ldots, \hat{X}_1 = a_1\big),
   \end{aligned}	
  \end{equation*}
  for any $1\leq i\leq \lceil \frac{n}{n'} \rceil$. From the chain rule of conditional probability, we obtain
  
   \begin{small}
   \begin{equation*}
  \begin{aligned}
     &\Pr(X_1 = a_1, X_2 = a_2, \ldots, X_{\lceil \frac{n}{n'} \rceil} = a_{\lceil \frac{n}{n'} \rceil}) \\
    =& \prod_{i=1}^{\lceil \frac{n}{n'} \rceil} \Pr\big(X_i = a_i \big|X_{i-1} = a_{i-1},X_{i-2} = a_{i-2}, \ldots, X_1 = a_1\big) \\
    =&\prod_{i=1}^{\lceil \frac{n}{n'} \rceil}\Pr\big(\hat{X}_i=a_i \big| \hat{X}_{i-1} = a_{i-1}, \hat{X}_{i-2} = a_{i-2},\ldots, \hat{X}_1 = a_1\big) \\
   = & \Pr(\hat{X}_1 = a_1, \hat{X}_2 = a_2, \ldots, \hat{X}_{\lceil \frac{n}{n'} \rceil} = a_{\lceil \frac{n}{n'} \rceil}). 
  \end{aligned}
  \end{equation*} 
\end{small}
  Finally, we return to the proof of Lemma \ref{lemma:permuted partition}. It can be obtained that 
   \begin{small}
    \begin{equation*}
    \begin{aligned}
      &\Pr(\#C(\vec{\lambda}) = h)  \\
    = &\sum _{(t_1,t_2,\ldots,t_{\lceil \frac{n}{n'} \rceil})\in L_h} \Pr(X_1 = t_1, X_2 = t_2, \ldots, X_{\lceil \frac{n}{n'} \rceil} = t_{\lceil \frac{n}{n'} \rceil})  \\
    \overset{(a)}{=} &\sum _{(t_1,t_2,\ldots,t_{\lceil \frac{n}{n'} \rceil})\in L_h} \Pr(\hat{X}_1 = t_1, \hat{X}_2 = t_2, \ldots, \hat{X}_{\lceil \frac{n}{n'} \rceil} = t_{\lceil \frac{n}{n'} \rceil}) \\
    \overset{(b)}{=}&\sum _{(\hat{t}_{1}, \hat{t}_{2}, \ldots, \hat{t}_{\lceil \frac{n}{n'} \rceil})\in \hat{L}_h}
     \Pr(\hat{X}_1 = \hat{t}_1, \hat{X}_2 = \hat{t}_2, \ldots, \hat{X}_{\lceil \frac{n}{n'} \rceil} = \hat{t}_{\lceil \frac{n}{n'} \rceil})  \\
    =&\Pr(\#C(\varphi(\vec{\lambda})) = h),
\end{aligned}
\end{equation*} 
\end{small}
 where $(a)$ and $(b)$ are derived from Claim \ref{claim:3} and Claim \ref{claim:2}, respectively. 
\end{proof}
By Lemma \ref{lemma:permuted partition}, for any valid partition of integer $t$, regardless of its internal permutation, the probability that the selected bits contain exactly $h$ $(1\leq h \leq k)$ distinct indices is determined. Thus, for convenience, 
we arrange the valid partition $\vec{\lambda}$ in descending order in subsequent discussions,
and the number of non-zero integers in $\vec{\lambda}$ is denoted by $\lceil \vec{\lambda} \rceil$. 

Next, we present an analysis of 
the light transmission and contrast resulting from the stacking of 
the selected share bits, as determined by $\vec{\lambda}$, 
and their corresponding shadow images. 

\begin{theorem}\label{theorem:t0t1_contrast}
  Suppose that $\lambda_i$ shadow images are selected from the $i$-th group to recover the secret image, where $1\leq i\leq \lceil \frac{n}{n'} \rceil$. Let $\vec{\lambda}=(\lambda_1,\lambda_2,\ldots,\lambda_{ \lceil \frac{n}{n'} \rceil})$ be a valid partition. Then, 
  \begin{enumerate}
  	\item the light transmission resulting from the superimposition of the bits indexed by $C(\vec{\lambda})$ is
  	  \begin{equation*}
  		\begin{aligned}
  			&t([b^0_{C(\vec{\lambda})}]_{\otimes}) = \frac{{\Pr}(\#C(\vec{\lambda}) = k)}{2^{k-1}} +
  			\sum_{g=1}^{k-1} \frac{{\Pr}(\#C(\vec{\lambda}) = g)}{2^g}, \\
  			&t([b^1_{C(\vec{\lambda})}]_{\otimes}) = \sum_{g=1}^{k-1} \frac{{\Pr}(\#C(\vec{\lambda}) = g)}{2^g} ,
  		\end{aligned}
  	\end{equation*}
  	 where $[b^s_{C(\vec{\lambda})}]_{\otimes}$ $(s = 0~ \text{or}~ 1)$ denotes the stacking result of the share bits corresponding to the indices in $C(\vec{\lambda})$;
  	\item the contrast of the recovered image is
  	  \begin{equation}\label{eq:single_layer_contrast}
  		\alpha 
  		= \frac{\frac{{\Pr}(\#C(\vec{\lambda}) = k)}{2^{k-1}}  }{1 + \sum_{g=1}^{k-1} \frac{{\Pr}(\#C(\vec{\lambda}) = g)}{2^g} }.  
  	\end{equation} 
  \end{enumerate}
\end{theorem}
\begin{proof}
	  \begin{enumerate}
 \item  From Lemma \ref{lemma:light transmission}, the value of $t([b^s_{C(\vec{\lambda})}]_{\otimes})$ is related to $\Pr(\#C(\vec{\lambda}) = k)$. The expected light transmission is: 
    \begin{equation*}
   	\begin{aligned}
  t([b^0_{C(\vec{\lambda})}]_{\otimes}) & = \sum_{g=1}^{k} {\Pr}(\#C(\vec{\lambda}) = g)\cdotp t(b'_g[s_{(0)}])  \\
                                 & =   \frac{{\Pr}(\#C(\vec{\lambda}) = k)}{2^{k-1}} +
                                 \sum_{g=1}^{k-1} \frac{{\Pr}(\#C(\vec{\lambda}) = g)}{2^g}, \\
             t([b^1_{C(\vec{\lambda})}]_{\otimes}) & = \sum_{g=1}^{k} {\Pr}(\#C(\vec{\lambda}) = g)\cdotp t(b'_g[s_{(1)}])  \\
           & =   \sum_{g=1}^{k-1} \frac{{\Pr}(\#C(\vec{\lambda}) = g)}{2^g}.\\                      
      	\end{aligned}
    \end{equation*}
\item 
 Due to Eq.~\eqref{eq:contrast_single_light_transmission}, we obtain
  \begin{small}
    \begin{equation*}
    \alpha = 
    \frac{t([b^0_{C(\vec{\lambda})}]_\otimes) - t([b^1_{C(\vec{\lambda})}]_\otimes)}{1 + t([b^1_{C(\vec{\lambda})}]_\otimes)}
    = \frac{\frac{{\Pr}(\#C(\vec{\lambda}) = k)}{2^{k-1}}  }{1 + \sum_{g=1}^{k-1} \frac{{\Pr}(\#C(\vec{\lambda}) = g)}{2^g} }.
    \end{equation*}
  \end{small}
 \end{enumerate}
\end{proof}

After calculating the contrast resulting from 
stacking the shadow images that correspond to the share bits selected by $\vec{\lambda}$, 
we can naturally proceed to compute the contrast achieved by 
stacking any $t~(1\leq t \leq n)$ shadow images. 
We first present an illustrative example, then formally define the contrast of 
$n'$-grouped $(k,n)$ RGVCS. 

\begin{example}
  Consider the case where $n' = 3, n = 7$, and $t = 3$.
  Due to Lemma \ref{lemma:permuted partition} and Eq.~\eqref{eq:single_layer_contrast}, 
  it can be seen that the contrast obtained by the valid partitions $(0, 2, 1)$, $(2, 0, 1)$, $(2, 1, 0)$, and $(1, 2, 0)$ is the same. 
  Therefore, we only need to consider the valid partition in descending order: $(3, 0, 0)$, $(2, 1, 0),$ and $(1, 1, 1)$. Let $\alpha_1$, $\alpha_2$, and $\alpha_3$ be the contrast resulting from stacking the shadow images that corresponding to  $(3, 0, 0)$, $(2, 1, 0),$ and $(1, 1, 1)$. Randomly select any $3$ of the $7$ shadow images. Through calculation, 
  it can be seen that the probabilities of contrast $\alpha_1$, $\alpha_2$, and $\alpha_3$ appearing are $\frac{2}{35}$, $\frac{24}{35}$, and $\frac{9}{35}$, respectively. Thus, 
  the contrast can be expressed as the expected contrast
  \begin{equation*}
    \alpha = \frac{2}{35}\times \alpha_1 + \frac{24}{35} \times \alpha_2 + \frac{9}{35} \times \alpha_3. 
  \end{equation*}
\end{example}

Due to Lemma \ref{lemma:permuted partition} and Theorem \ref{theorem:t0t1_contrast}, 
each valid partition, regardless of its internal permutation, corresponds to a unique contrast value. Thus, the expected contrast can be used to represent the contrast of the $n'$-grouped $(k,n)$ RGVCS. We formally define the contrast as follows.

\begin{definition}
  \label{definition:grouped contarst}
  Let $Z$ be the set of all valid partitions of $t~(1\leq t \leq n)$, 
  which can be expressed as: 
  \begin{equation*}
    Z \triangleq \bigsqcup_{\vec{\lambda}\in Z}[\vec{\lambda}], 
  \end{equation*}
  where each equivalence class $[\vec{\lambda}]$ is defined as: 
  \begin{equation*}
    [\vec{\lambda}] = \{\varphi(\vec{\lambda})\in Z ~|~ \varphi \mbox{ is a permutation}\}.
  \end{equation*} 
  Suppose $Z$ contains $g$ equivalence classes, where the contrast values of each equivalence class are sorted in descending order, and denote them as $\alpha_1,\alpha_2,\ldots,\alpha_g$. 
  Let $\beta_1,\beta_2,\ldots,\beta_g$ denote the probabilities of occurrence for the equivalence classes 
  corresponding to $\alpha_1,\alpha_2,\ldots,\alpha_g$ in $Z$. 
  Then, the contrast of the $n'$-grouped $(k,n)$ RGVCS, denoted as $\Gamma$, is defined as
 $\Gamma \triangleq \sum_{i=1}^{g}\beta_i \alpha_i=\beta_1 \alpha_1+\beta_2 \alpha_2+\cdots+\beta_{g} \alpha_{g}$.
 \end{definition}

\subsection{The validity of the proposed scheme}

At the end of this section, 
we demonstrate the validity of the $n^{\prime}$-grouped $(k,n)$ RGVCS. 
\begin{theorem} 
  \label{theorem:n' scheme validity}
  The $n^{\prime}$-grouped $(k,n)$ RGVCS is a valid scheme.
\end{theorem}
\begin{proof}
  \begin{enumerate}
    \item We first prove the visually recognizable condition. 
  Consider selecting any $k$ shadow images for recovery. 
  For any valid partition of $k$, denoted as $\vec{\lambda}$, 
  the probability of occurrence for its corresponding equivalence class in $Z$ must be greater than 0, 
  which means $\beta_1,\beta_2,\ldots,\beta_g$ are greater than 0. 
  Additionally, it's easy to obtain that ${\Pr}(\#C(\vec{\lambda}) = k)>0$ by definition, 
  so the numerator in Eq.~\eqref{eq:single_layer_contrast} is greater than 0, 
  ensuring that the corresponding contrast is greater than 0. 
  That is to say, $\alpha_1,\alpha_2,\ldots,\alpha_g$ are greater than 0. 
  Therefore, $\Gamma= \sum_{i=1}^{g}\beta_i \alpha_i>0$. 
  \item  Next, we prove the security condition. 
  Consider selecting any $q~(1\leq q < k)$ shadow images for recovery. 
  Let $\vec{\zeta}$ be one of the valid partition of $q$. 
  From definition, we obtain ${\Pr}(\#C(\vec{\lambda}) = k)=0$.
  Thus, the numerator in Eq.~\eqref{eq:single_layer_contrast} is $0$ and the corresponding contrast is 0. Then, $\alpha_1,\alpha_2,\ldots,\alpha_g$ are all equal to 0. 
  Therefore, $\Gamma= 0$. 
  \end{enumerate}
\end{proof}
In summary, we propose a novel $(k,n)$ RGVCS based on the new bit-level sharing paradigm,  
along with a brand-new calculation method for contrast.   
We specifically discuss the case where $n'$ is set to $k$ in Section \ref{section n'=k}. 

\section{$k$-grouped $(k,n)$ RGVCS}{\label{section n'=k}}
Since only the $(k,k)$ RGVCS is currently known to 
achieve the upper bound of contrast\cite{14TIFS}, 
and existing $(k,n)$ RGVCSs mentioned in Section \ref{section:Preliminaries} 
are all extensions based on the $(k,k)$ RGVCS, 
in this section, we fix the adjustable value $n^{\prime}$ to $k$, 
propose a $k$-grouped $(k,n)$ RGVCS. 
We first describe the specific algorithm of the scheme 
and then analyze the performance. 

\begin{algorithm}[t]
  \caption{$k$-grouped $(k,n)$ RGVCS}
  \label{algorithm:grouped k_k_n}
  \SetKwData{In}{\textbf{in}}\SetKwData{To}{to}
  \DontPrintSemicolon
  \SetAlgoLined
  \KwIn {an $h\times w$ binary secret image $S$;\\threhold parameters $k,n$, where $k \leq n$}
  \KwOut {$SC_1,SC_2,\ldots,SC_{n}$}
  \Begin{
      \For{$(i,j)$ where $1\leq i \leq h ,1\leq j \leq w$}{
          Execute $(S[i,j],k,k)$ RGVCS to obtain $b_1,b_2,\ldots,b_{k}$ 
          and distribute to $SC_1[i,j],SC_2[i,j],\ldots,SC_{k}[i,j]$ \;
          \For{$t \gets k+1$ \To $n$}{
         
            \If{$t\;\mod\;k=1$}{
                $\mathcal{P}=\emptyset$
            }
            $p = T_{\{1,2,\ldots,k\}\setminus \mathcal{P}}$\;
            $SC_t[i,j] = b_p$\;
            $\mathcal{P}.append(p)$\;            
          } 
      }
      \Return{$SC_1,SC_2,\ldots,SC_{n}$}  
  }
\end{algorithm}

\subsection{The sharing and recovery phase}
Algorithm \ref{algorithm:grouped k_k_n} describes the specific steps of the sharing phase of the $k$-grouped $(k,n)$ RGVCS. 
By comparing with Algorithm \ref{algorithm:n' grouped k_n}, 
the most significant difference is the specific RGVCS executed in step $3$.
Here, the first share bits group $\mathcal{G}_1$ is generated by executing $(s,k,k)$ RGVCS. 
Similar to Algorithm \ref{algorithm:n' grouped k_n}, 
the share bits and shadow images are organized into groups of $k$.  
As for the recovery phase, the same method of superimposition is used. 

\subsection{Performance analysis}
This section primarily discusses the contrast of the $k$-grouped $(k,n)$ RGVCS.
Since this scheme is a special case of the $n^{\prime}$-grouped $(k,n)$ RGVCS, 
many symbolic representations and conclusions in Section \ref{section:grouped k_n'_n} can be directly applied. 
In $k$-grouped $(k,n)$ RGVCS, 
the calculation of contrast follows Definition~\ref{definition:grouped contarst},
with the difference lying in the specific values of $\alpha_i$ and $\beta_i$. 
In this subsection, 
we primarily analyze the computational methods for $\alpha_i$ and $\beta_i$ under $k$-grouped $(k,n)$ RGVCS. 


We first introduce the calculation of ${\Pr}(\#C(\vec{\lambda}) = x)~(1\leq x \leq k)$ in Eq.~\eqref{eq:single_layer_contrast}. 
We define the compliant matrix to simplify the subsequent exposition as follows. 

\begin{definition}(Compliant Matrix)
	\label{definition:compliant matrix}
	Let the number of distinct elements in $C(\vec{\lambda})$ represent as $x$, 
	and denote them as $n_1,n_2,\ldots,n_x$. 
	Consider the sequence $n_1,n_2,\ldots,n_x$ arranged as a matrix of size $\lceil \frac{n}{k} \rceil \times x$. 
	This matrix is called a compliant matrix if and only if it satisfies the following row and column constraints: 
	\begin{enumerate}
		\item each element of the matrix takes a value of $0$ or $1$;  
		\item the sum of the elements in each row corresponds to each integer in $\vec{\lambda}$, 
		and the sum of the elements in each column is at least $1$.
	\end{enumerate}
\end{definition} 

Let $\mathcal{E}_{x}^{\vec{\lambda}}(B)$ be the set of compliant matrices whose last row is fixed to $B$, where $B\in\{0,1\}^{1\times x}$ with Hamming weight $\lambda_{\lceil \frac{n}{k} \rceil}$, which can be expressed as

\begin{footnotesize}
	\begin{equation*}
		\label{eq:E}
		\begin{aligned}
			&\mathcal{E}_{x}^{\vec{\lambda}}(B) = \\
			\!\!\!\!&\left\{ E=\begin{pmatrix} A \\ B \end{pmatrix} \in \{0,1\}^{\lceil \frac{n}{k} \rceil \times x} \left|
			\begin{aligned}
				&\sum_{j=1}^{x} E(i, j) = \lambda_i, i \in \{1,2,\ldots,\lceil \frac{n}{k} \rceil\},  \\
				&\sum_{i=1}^{\lceil \frac{n}{k} \rceil} E(i, j)\geq 1, j \in \{1,2,\ldots,x\}  \\
			\end{aligned}
			\right.
			\right\}, 
		\end{aligned}
	\end{equation*}
\end{footnotesize}
where $E(i,j)$ denotes the element at the $i$-th row and $j$-th column of $E$.
Furthermore, we present the following example to clarify Definition \ref{definition:compliant matrix}. 

\begin{example}
	Consider the case when $k = 6$, $t = 6$, and $n = 14$, for a given valid partition $(3,2,1)$, 
	the following two examples $E_1$ and $E_2$ illustrate matrices that meet the constraints for $x=4$. 
	The selected index multiset 
	$C((3,2,1))$ corresponding to $E_1$ is $({n_1},{n_2},{n_3},{n_3},{n_4},{n_3})$,  
	and for $E_2$ it is $({n_1},{n_2},{n_3},{n_2},{n_4},{n_4})$.
	\[
	E_1 =
	\begin{pmatrix}
		1 & 1 & 1 & 0 \\
		0 & 0 & 1 & 1 \\
		0 & 0 & 1 & 0 \\
	\end{pmatrix},
	\quad
	E_2 = 
	\begin{pmatrix}
		1 & 1 & 1 & 0 \\
		0 & 1 & 0 & 1 \\
		0 & 0 & 0 & 1 \\
	\end{pmatrix}. 
	\]
\end{example}

\begin{algorithm}[t]
	\caption{Compliant matrix counting}
	\label{algorithm:compliant matrix counting}
	\SetKwData{In}{\textbf{in}}\SetKwData{To}{to}
	\DontPrintSemicolon
	\SetAlgoLined
	\KwIn {valid partition $\vec{\lambda}$, parameters $x$, $B$}
	\KwOut {$|\mathcal{E}^{\vec{\lambda}}_x(B)|$}
	\Begin{
		Initialize a $\lceil \frac{n}{k} \rceil \times x$ all-zeros matrix $\mathcal{M}_{\lceil \frac{n}{k} \rceil \times x}$\;
		$\mathcal{E}^{\vec{\lambda}}_x(B)$ = Backtrack($\mathcal{M}_{\lceil \frac{n}{k} \rceil \times x},\vec{\lambda},0,0,\emptyset,B$)\;
		
		\Return{$|\mathcal{E}^{\vec{\lambda}}_x(B)|$}  
	}
\end{algorithm}   
\begin{algorithm}[t]\caption{Backtrack($\mathcal{M}_{\lceil \frac{n}{k} \rceil \times x}, \vec{\lambda},r,c,R,B$) \protect\footnotemark}
	
	\label{algorithm:Backtracking function}
	\SetKwData{In}{\textbf{in}}\SetKwData{To}{to}
	\DontPrintSemicolon
	\SetAlgoLined
	\KwIn {a matrix to be filled $\mathcal{M}_{\lceil \frac{n}{k} \rceil \times x}$, $\vec{\lambda}$, current row $r$ and column $c$, solution set $R$, vector $B$}
	\KwOut {solution set $R$}
	\Begin{
		Set the last row of $\mathcal{M}$ according to vector $B$\;
		\If{$r\geq\lceil \frac{n}{k} \rceil$ }{  
			\If{the number of $1$s in all rows is equal to $\lambda_i$ and each column contains at least one $1$}{
				Add the current matrix $\mathcal{M}$ to $R$
			}
		}
		\If{$r=\lceil \frac{n}{k} \rceil-1$}{
			Backtrack($\mathcal{M}_{\lceil \frac{n}{k} \rceil \times x}, \vec{\lambda},r+1,0,R,B$)
		}
		\If{$c\geq x$}{ 
			\If{the number of $1$s in current row is equal to $\lambda_r$}{
				Backtrack($\mathcal{M}_{\lceil \frac{n}{k} \rceil \times x}, \vec{\lambda},r+1,0,R,B$) 
			}
		}
		\For{$v \in \{1,0\}$}{
			$\mathcal{M}[r][c]\gets v$\;
			\If{$v = 1$ and the number of 1s in current row is greater than $\lambda_r$}{
				continue
			}
			Backtrack($\mathcal{M}_{\lceil \frac{n}{k} \rceil \times x}, \vec{\lambda},r,c+1,R,B$)
		}

		
		\Return{$R$}  
	}
\end{algorithm}
\footnotetext{Algorithm \ref{algorithm:Backtracking function} is a backtracking algorithm.}

Next, we derive ${\Pr}(\#C(\vec{\lambda}) = x)$ through the count of compliant matrices, 
i.e., $|\mathcal{E}^{\vec{\lambda}}_x(B)|$, 
whose calculation is detailed in Algorithm \ref{algorithm:compliant matrix counting} and Algorithm \ref{algorithm:Backtracking function}.
\begin{lemma} \label{theorem:pr}
	Suppose that $\lambda_i$ shadow images are selected from the $i$-th group to recover the secret image, where $1\leq i\leq \lceil \frac{n}{n'} \rceil$. Let $\vec{\lambda}=(\lambda_1,\lambda_2,\ldots,\lambda_{ \lceil \frac{n}{k} \rceil})$ be a valid partition. Then,  
	\begin{equation}
		\label{eq:Pr2}
		{\Pr}(\#C(\vec{\lambda}) = x) = \frac{\binom{k-\lambda_{\lceil \frac{n}{k} \rceil}}{x - \lambda_{\lceil \frac{n}{k}\rceil}} |\mathcal{E}^{\vec{\lambda}}_x(B)|}
		{  \prod_{j=1}^{\lceil \frac{n}{k} \rceil -1}\binom{k}{\lambda_j}}. 
	\end{equation} 
\end{lemma}
\begin{proof}
	Without considering the value of $\#C(\vec{\lambda})$, there are 
	$\prod_{j=1}^{\lceil \frac{n}{k} \rceil}\binom{|\mathcal{G}_{j}|}{\lambda_j}=\binom{|\mathcal{G}_{\lceil \frac{n}{k} \rceil}|}{\lambda_{\lceil \frac{n}{k} \rceil}}\prod_{j=1}^{\lceil \frac{n}{k} \rceil -1}\binom{k}{\lambda_j}$
	possible ways to select.
	Next, we consider how many selection ways there are when $\#C(\vec{\lambda})=x$.
	Consider this issue in two steps. 
	First, we select $\lambda_{\lceil \frac{n}{k} \rceil}$ distinct bit indices from $\mathcal{G}_{\lceil \frac{n}{k} \rceil}$ and determine $x$ distinct bit indices.
	There are $\binom{|\mathcal{G}_{\lceil \frac{n}{k} \rceil}|}{\lambda_{\lceil \frac{n}{k} \rceil}} \binom{k-\lambda_{\lceil \frac{n}{k} \rceil}}{x - \lambda_{\lceil \frac{n}{k}\rceil}}$ possible ways to select.
	Second, one way to select the last group of indices and $x$ distinct bit indices, corresponding to various selection ways for $\#C(\vec{\lambda})=x$, each of which corresponds to a matrix in $\mathcal{E}^{\vec{\lambda}}_x(B)$.
	Thus, there are $\binom{|\mathcal{G}_{\lceil \frac{n}{k} \rceil}|}{\lambda_{\lceil \frac{n}{k} \rceil}} \binom{k-\lambda_{\lceil \frac{n}{k} \rceil}}{x - \lambda_{\lceil \frac{n}{k}\rceil}}|\mathcal{E}^{\vec{\lambda}}_x(B)|$ 
	selection ways when $\#C(\vec{\lambda})=x$. Therefore, we obtain 
	\begin{small}
		\begin{equation*}
			\begin{aligned}
				{\Pr}(\#C(\vec{\lambda}) = x) &= \frac{\binom{|\mathcal{G}_{\lceil \frac{n}{k} \rceil}|}{\lambda_{\lceil \frac{n}{k} \rceil}} \binom{k-\lambda_{\lceil \frac{n}{k} \rceil}}{x - \lambda_{\lceil \frac{n}{k}\rceil}} |\mathcal{E}^{\vec{\lambda}}_x(B)|}
				{ \binom{|\mathcal{G}_{\lceil \frac{n}{k} \rceil}|}{\lambda_{\lceil \frac{n}{k} \rceil}} \prod_{j=1}^{\lceil \frac{n}{k} \rceil -1}\binom{k}{\lambda_j}}\\\\
				&=\frac{\binom{k-\lambda_{\lceil \frac{n}{k} \rceil}}{x - \lambda_{\lceil \frac{n}{k}\rceil}} |\mathcal{E}^{\vec{\lambda}}_x(B)|}
				{  \prod_{j=1}^{\lceil \frac{n}{k} \rceil -1}\binom{k}{\lambda_j}}.
			\end{aligned}
		\end{equation*}
	\end{small}
\end{proof}

Then, we describe the calculation of $\beta_i$. 

\begin{lemma}
	\label{theorem:betai}
	Suppose that $\lambda_i$ shadow images are selected from the $i$-th group to recover the secret image, 
	where $1\leq i\leq \lceil \frac{n}{k} \rceil$. 
	Let $t\triangleq \sum_{i=1}^{\lceil\frac{n}{k}\rceil} \lambda_i$. Let $\vec{\lambda}=(\lambda_1,\lambda_2,\ldots,\lambda_{ \lceil \frac{n}{k} \rceil})$ 
	be a valid partition, which can be expressed in the following form:
	$$\vec{\lambda} = \{{l_1}^{\!\!d_1},{l_2}^{\!\!d_2},\ldots,{l_r}^{\!\!d_r}\},$$
	where $l_1,l_2,\ldots,l_r$ denote $r$ distinct integers among the components of
	$\vec{\lambda}$, and $d_j~(1\leq j \leq r)$ denotes the number of occurrences of $l_j$. 
	Let $\beta_{\vec{\lambda}}$ denote the occurrence probability of the equivalence class $[\vec{\lambda}]$ in $Z$.\footnote{The symbols are as defined in Definition~\ref{definition:grouped contarst}.}
	Then, 
	\begin{small}
	\begin{equation}
		\label{eq:beta}
		\beta_{\vec{\lambda}} = \frac{ \sum_{g=1}^{r} \left[\binom{|\mathcal{G}_{\lceil \frac{n}{k} \rceil}|}{l_g} \prod_{j=1}^{\lceil \frac{n}{k} \rceil-1}\binom{k}{\lambda_j} \frac{(\lceil \frac{n}{k} \rceil-1)!}{d_1!\cdots d_{g-1}!(d_g-1)!d_{g+1}!\cdots d_r!}\right]}{\binom{n}{t}}.
	\end{equation}
	\end{small}
\end{lemma}
\begin{proof}
	There are $r$ possible values of $\lambda_{ \lceil \frac{n}{k} \rceil}:l_1,l_2,\ldots,l_r$.
	When $l_g$ $(1\leq g \leq r)$ shadow images is chosen from $\mathcal{G}_{\lceil \frac{n}{k} \rceil}$, 
	and the remaining $\lceil \frac{n}{k} \rceil-1$ groups of shadow images are all selected from the previous groups with size of $k$. 
	Thus, the number of selection ways is 
	$\prod_{j=1}^{\lceil \frac{n}{k} \rceil}\binom{|\mathcal{G}_{j}|}{\lambda_j} =\binom{|\mathcal{G}_{\lceil \frac{n}{k} \rceil}|}{l_g} \prod_{j=1}^{\lceil \frac{n}{k} \rceil-1}\binom{k}{\lambda_j}$. 
	There are $Y$ valid partitions in equivalence class $[\vec{\lambda}]$ that are equivalent to $\vec{\lambda}$, where $Y$ is equal to the number of permutations $\frac{(\lceil \frac{n}{k} \rceil-1)!}{d_1!\cdots d_{g-1}!(d_g-1)!d_{g+1}!\cdots d_r!}$ of the multiset $\{{l_1}^{\!\!d_1},\ldots,{l_{g-1}}^{\!\!d_{g-1}},{l_g}^{\!\!d_g-1},{l_{g+1}}^{\!\!d_{g+1}},\ldots,{l_r}^{\!\!d_r}\}$. Additionally, there are $\binom{n}{t}$ ways to select $t$ shadow images from $n$. Thus, Eq.~\eqref{eq:beta} holds.
\end{proof}

\begin{remark}
  When $|\mathcal{G}_{\lceil \frac{n}{k} \rceil}| < l_g$, 
  the value of $\binom{|\mathcal{G}_{\lceil \frac{n}{k} \rceil}|}{l_g}$ is 0. For instance, 
  we cannot select 4 shadow images from a set containing only 3 shadow images. 
\end{remark}

So far, we have calculated all the parameters related to the contrast of the ${k}$-grouped $(k,n)$ RGVCS. 
To enhance clarity for readers, 
we provide a detailed computation of contrast for $t=3$ under $3$-grouped $(3,7)$-threshold. 

\begin{example}
	The valid partitions of 3 are: $(3,0,0)$, $(2,1,0)$, and $(1,1,1)$. 
	Their compliant matrices are shown below: 
	\begin{footnotesize}
		\[
		\begin{aligned}
			&\mathcal{E}_3^{(3,0,0)}\!(B_1) \!= \!
			\left\{\!\!
			\begin{pmatrix}
				1 & 1 & 1 \\
				0 & 0 & 0 \\
				0 & 0 & 0 \\
			\end{pmatrix}\!\!
			\right\}\!; \notag
			\mathcal{E}_2^{(2,1,0)}\!(B_2) \!= \!
			\left\{\!\!
			\begin{pmatrix}
				1 & 1 \\
				1 & 0 \\
				0 & 0 \\
			\end{pmatrix},
			\begin{pmatrix}
				1 & 1 \\
				0 & 1 \\
				0 & 0 \\
			\end{pmatrix}\!\!
			\right\}, \notag\\
			&\mathcal{E}_3^{(2,1,0)}(B_1) \!= \!
			\left\{\!\!
			\begin{pmatrix}
				1 & 1 & 0\\
				0 & 0 & 1\\
				0 & 0 & 0\\
			\end{pmatrix},
			\begin{pmatrix}
				1 & 0 & 1\\
				0 & 1 & 0\\
				0 & 0 & 0\\
			\end{pmatrix},
			\begin{pmatrix}
				0 & 1 & 1\\
				1 & 0 & 0\\
				0 & 0 & 0\\
			\end{pmatrix}\!\!
			\right\}; \notag\\
			&\mathcal{E}_1^{(1,1,1)}\!(B_3) \!= \!
			\left\{\!\!
			\begin{pmatrix}
				1 \\
				1 \\
				1 \\
			\end{pmatrix}\!\!
			\right\}\!, \!
			\mathcal{E}_2^{(1,1,1)}\!(B_4) \!= \!
			\left\{\!\!
			\begin{pmatrix}
				1 & 0\\
				0 & 1\\
				1 & 0\\
			\end{pmatrix}\!,\!
			\begin{pmatrix}
				0 & 1\\
				1 & 0\\
				1 & 0\\
			\end{pmatrix}\!,\!
			\begin{pmatrix}
				0 & 1\\
				0 & 1\\
				1 & 0\\
			\end{pmatrix}\!\!
			\right\}\!, \!\notag\\
			&\mathcal{E}_3^{(1,1,1)}(B_5) \!= \!
			\left\{\!\!
			\begin{pmatrix}
				0 & 1 & 0\\
				0 & 0 & 1\\
				1 & 0 & 0\\
			\end{pmatrix},
			\begin{pmatrix}
				0 & 0 & 1\\
				0 & 1 & 0\\
				1 & 0 & 0\\
			\end{pmatrix}\!\!
			\right\},\notag
		\end{aligned}
		\]
	\end{footnotesize}
	where $B_1=(0,0,0)$, $B_2=(0,0)$, $B_3=(1)$, $B_4=(1,0)$ and $B_5=(1,0,0)$.
	Thus, it can be calculated from Eq.(\ref{eq:Pr2}) that:  
	$\Pr(\#C((3,0,0))=3) = 1$; 
	$\Pr(\#C((2,1,0))=2) = \frac{2}{3}$, $\Pr(\#C((2,1,0))=3) = \frac{1}{3}$; 
	$\Pr(\#C((1,1,1))=1) = \frac{1}{9}$, $\Pr(\#C((1,1,1))=2) = \frac{2}{3}$, $\Pr(\#C((1,1,1))=3) = \frac{2}{9}$. 
	
	Therefore, when stacking 3 shadow images under $(3,7)$-threshold, 
	the contrast is primarily divided into three levels: 
	
	\begin{footnotesize}
		\begin{align*}
			\alpha_1 &= \frac{{\Pr}(\#C(\vec{\lambda})=3)\times (\frac{1}{2})^{3-1}}{1+0} = \frac{1\times (\frac{1}{2})^2}{1+0} = \frac{1}{4},\\
			\alpha_2 &= \frac{{\Pr}(\#C(\vec{\lambda})=3)\times (\frac{1}{2})^{3-1}}{1+{\Pr}(\#C(\vec{\lambda})=2)\times (\frac{1}{2})^2} = \frac{\frac{1}{3}\times (\frac{1}{2})^2 }{1+\frac{2}{3}\times (\frac{1}{2})^2} = \frac{1}{14},\\
			\alpha_3 &= \frac{{\Pr}(\#C(\vec{\lambda})=3)\times (\frac{1}{2})^{3-1}}{1+{\Pr}(\#C(\vec{\lambda})=1)\times (\frac{1}{2}) + {\Pr}(\#C(\vec{\lambda})=2)\times (\frac{1}{2})^2}  \\
			&= \frac{\frac{2}{9}\times (\frac{1}{2})^2 }{1+\frac{1}{9}\times (\frac{1}{2}) + \frac{2}{3}\times (\frac{1}{2})^2} 
			= \frac{1}{22}.
		\end{align*}   
	\end{footnotesize}
	Moreover, the calculation of $\beta_i$ can be derived from Eq.~\eqref{eq:beta}, 
	which are $\beta_1 =  \frac{2}{35}$, $\beta_2 = \frac{24}{35}$, 
	and $\beta_3 = \frac{9}{35}$. 
	Thus, the contrast is 
	
	\begin{footnotesize}
		\begin{equation*}
			\Gamma = \alpha_1\beta_1 + \alpha_2\beta_2 + \alpha_3\beta_3 
			= \frac{1}{4}\times \frac{2}{35} + \frac{1}{14}\times \frac{24}{35} + \frac{1}{22}\times \frac{9}{35} 
			= \frac{202}{2695}.
		\end{equation*}
	\end{footnotesize}
	
\end{example}

In addition, 
we present lists of the values of contrast under  
$(3,5)$-threshold and $(4,5)$-threshold, 
which are demonstrated in Table \ref{tab:k=3,k=4,n=5}. 

\section{Experiments and Comparisons}\label{section:experiment}
In the study of VCS, 
contrast is the key metric to measure the schemes.
In this section, we experimentally evaluate the performance of our proposed $k$-grouped $(k, n)$ RGVCS and conduct a comparative analysis with existing mainstream schemes to validate its advantages.


\subsection{Image illustration}

In this subsection, we conduct experiments of $(3,7)$-threshold on the proposed $k$-grouped $(k,n)$ RGVCS and Shyu's RGVCS. 
The recovered images and specific contrast values are presented in Figure \ref{fig:927320} and Table \ref{tab:927320}. 
The experimental results demonstrate that, 
under the $(3,7)$-threshold, our proposed scheme achieves three-level recovery effects, 
compared to Shyu's scheme, 
showing superior recovery performance at the highest level, 
maintaining comparable results at the secondary level,  
and exhibiting weaker performance at the lowest level. 

\begin{figure}[t]
  \centering
  \subfloat[]{
    \includegraphics[width=0.2\textwidth]{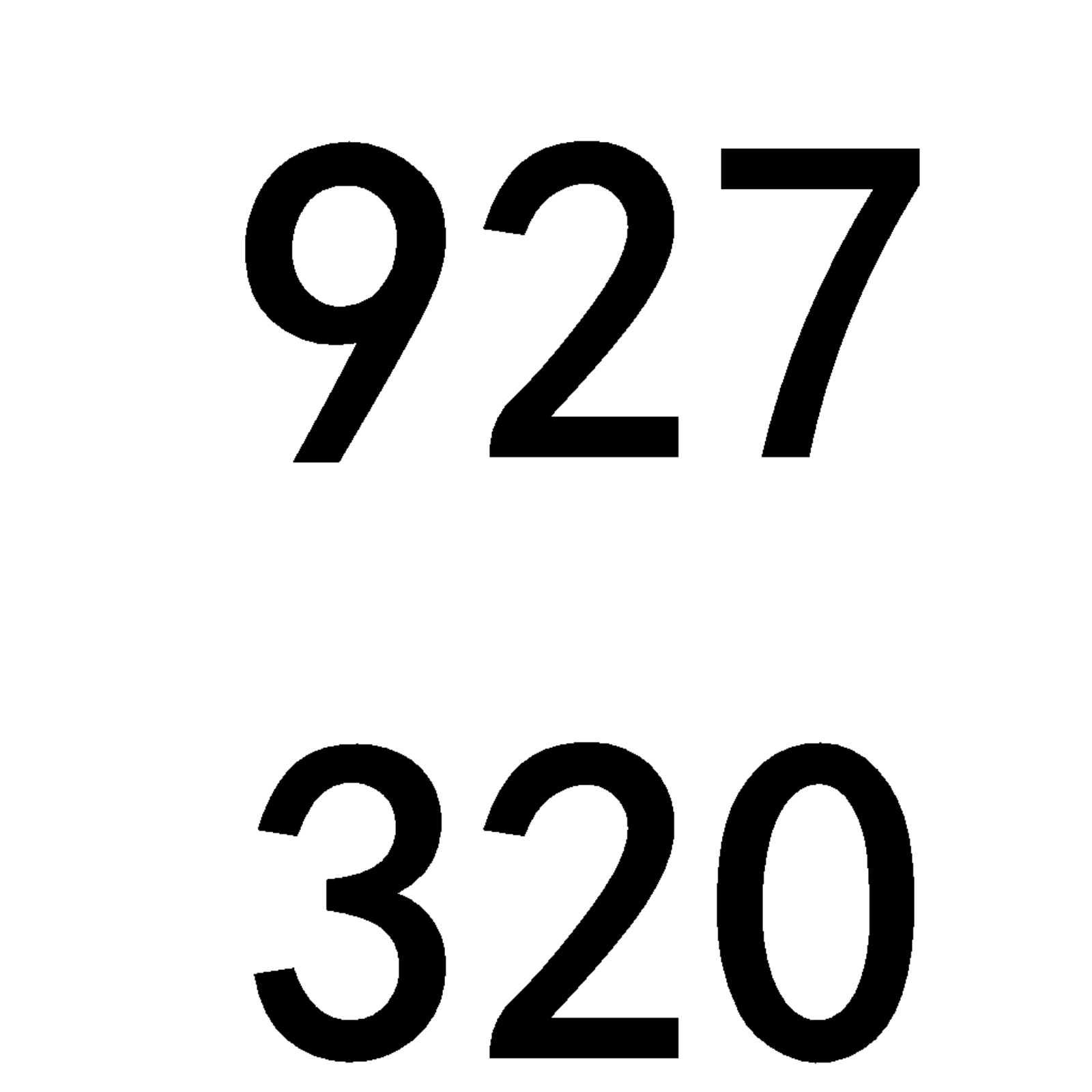}}
  \subfloat[]{
      \includegraphics[width=0.2\textwidth]{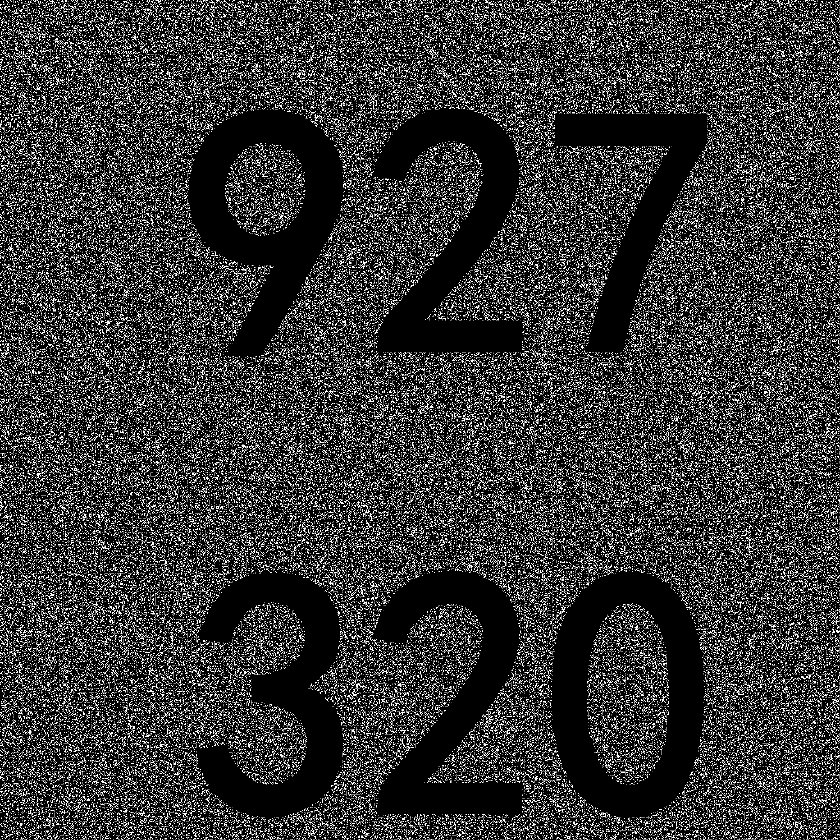}}
  \subfloat[]{
      \includegraphics[width=0.2\textwidth]{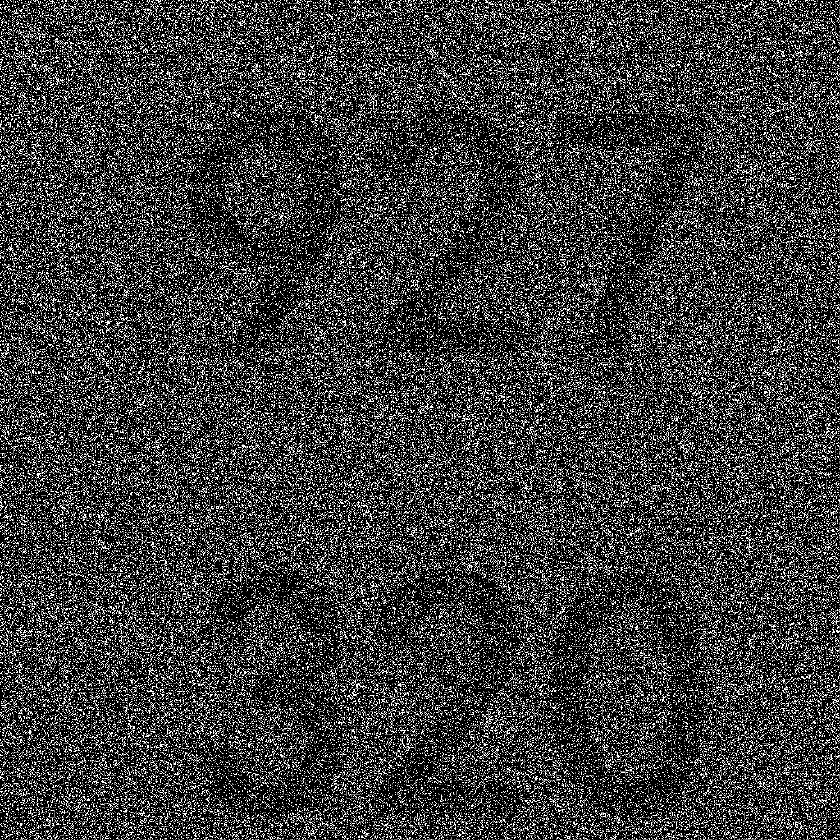}}
  \subfloat[]{
      \includegraphics[width=0.2\textwidth]{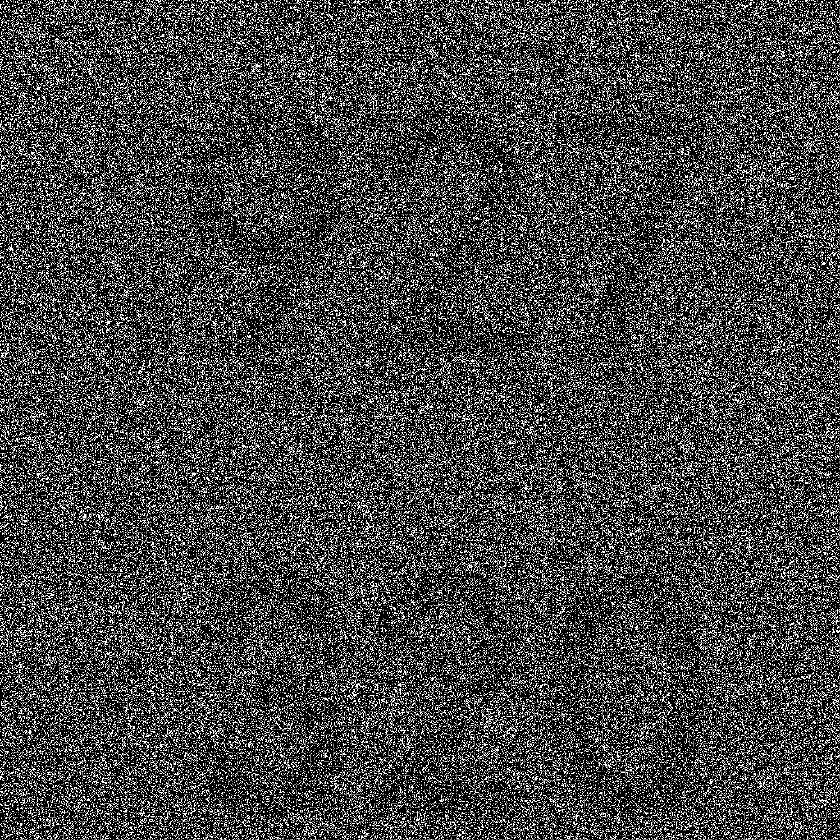}}
  \subfloat[]{
      \includegraphics[width=0.2\textwidth]{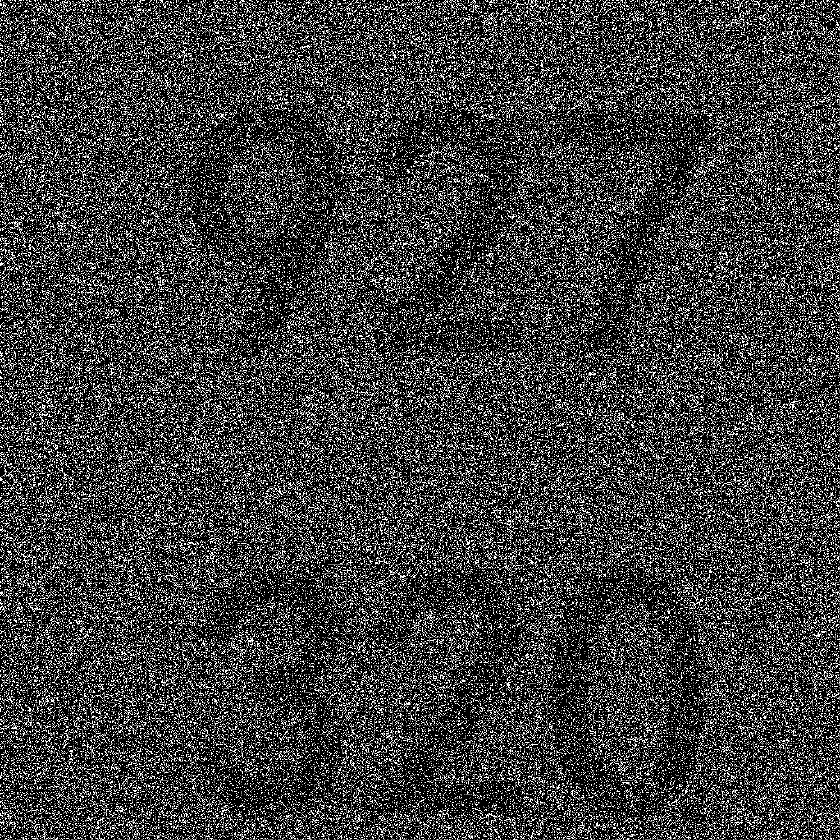}}
  \caption{The secret image and recovered images under $(3,7)$-threhold. 
  (a) The secret image; 
  (b)-(d) Three distinct recovery effects of our scheme; 
  (e) The recovery effect of Shyu's scheme.}
  \label{fig:927320}
  \vspace{-8pt} 
  \end{figure}
  \begin{table}[t] 
    \centering
    \caption{Comparison of theoretical and experimental contrast values under $(3,7)$-threshold}
    \label{tab:927320}
    \begin{adjustbox}{width=0.5\columnwidth}
      \small
      {\fontsize{4}{5}\selectfont
\begin{tabular}{ccc}
      \hline
                     & Experimental contrast & Theoretical contrast \\ \hline
      (b) & 0.2495                & 1/4                  \\
      (c) & 0.0706                & 1/14                 \\
      (d) & 0.0437                & 1/22                 \\ 
      (e)     & 0.0729                & 3/41                 \\ \hline
      \end{tabular}
      }
    \end{adjustbox}
    \vspace{-10pt} 
    \end{table}

\subsection{Comparisons with other schemes}
\begin{table*}
	\centering
	\caption{The contrast calculation for $k=3,n=5$ and $k=4,n=5$}\label{tab:k=3,k=4,n=5}
	\resizebox{\linewidth}{!}{
		\begin{tabular}{ccccccccccccccccc}
			\hline
			&  & \multicolumn{15}{c}{n = 5}                                                                                                                                                                              \\ \cline{3-11} \cline{13-17} 
			&  & \multicolumn{9}{c}{k = 3}                                                                                                    &  & \multicolumn{5}{c}{k = 4}                                             \\
			&  & \multicolumn{3}{c}{t = 3}                          &  & \multicolumn{3}{c}{t = 4}                        &  & t = 5          &  & \multicolumn{3}{c}{t = 4}                         &  & t = 5          \\ \hline
			i                          &  & 1             & \multicolumn{2}{c}{2}              &  & 1              & \multicolumn{2}{c}{2}           &  & 1              &  & 1            & \multicolumn{2}{c}{2}              &  & 1              \\
			valid partitions         &  & (3, 0)        & \multicolumn{2}{c}{(2, 1)}         &  & (3, 1)         & \multicolumn{2}{c}{(2, 2)}      &  & (3, 2)         &  & (4, 0)       & \multicolumn{2}{c}{(3, 1)}         &  & (4, 1)          \\
			$\{{l_1}^{\!\!d_1},{l_2}^{\!\!d_2},\ldots,{l_r}^{\!\!d_r}\}$              &  & $\{3^1,0^1\}$     & \multicolumn{2}{c}{$\{2^1, 1^1\}$} &  & $\{3^1, 1^1\}$ & \multicolumn{2}{c}{$\{2^2\}$}   &  & $\{3^1, 2^1\}$ &  & $\{4^1,0^1\}$    & \multicolumn{2}{c}{$\{3^1, 1^1\}$} &  & $\{4^1, 1^1\}$ \\
			x                          &  & 3             & 2                & 3               &  & 3              & 2               & 3             &  & 3              &  & 4            & 3                 & 4              &  & 4              \\
			$|\mathcal{E}^{\vec{\lambda}}_x(B)|$                      &  & 1             & 1                & 1               &  & 1              & 1               & 2             &  & 1              &  & 1            & 1                 & 1              &  & 1              \\
			$\Pr(\#C(\vec{\lambda}=x))$ &  & 1             & 2/3              & 1/3             &  & 1              & 1/3             & 2/3           &  & 1              &  & 1            & 3/4               & 1/4            &  & 1              \\
			$\alpha_i$                 &  & 1/4           & \multicolumn{2}{c}{1/14}           &  & 1/4            & \multicolumn{2}{c}{2/13}        &  & 1/4            &  & 1/8          & \multicolumn{2}{c}{1/35}           &  & 1/8            \\
			$\beta_i$                  &  & 1/10 & \multicolumn{2}{c}{9/10}  &  & 2/5    & \multicolumn{2}{c}{3/5} &  & 1        &  & 1/5  & \multicolumn{2}{c}{4/5}    &  & 1        \\\cline{3-11} \cline{13-17}
			$\Gamma$                   &  & \multicolumn{3}{c}{5/56}                           &  & \multicolumn{3}{c}{5/26}                         &  & 1/4            &  & \multicolumn{3}{c}{67/1400}                       &  & 1/8            \\ \hline
		\end{tabular}
	}
	\vspace{-10pt} 
\end{table*}
We conduct a comprehensive comparison of our $k$-grouped $(k,n)$ RGVCS 
with current mainstream schemes 
in terms of theoretical and experimental contrast, 
with the results summarized in Tables \ref{tab:theoretical contrast comparisons} and \ref{tab:experimental contrast comparisons}, 
where the best contrasts are highlighted in bold.
Additionally, 
Figure \ref{fig:Ours_Shyu_Yan} illustrates the contrast curves associated with 
the recovery of different participant combinations under the $(3,12)$-threshold for our scheme, Shyu's scheme,  
and Yan's scheme, respectively. 
The analysis allows us to draw the following conclusions.

  $\bullet$ 
  From the data presented in Tables \ref{tab:theoretical contrast comparisons} and \ref{tab:experimental contrast comparisons}, 
  the experimental values are consistent with the theoretical values. 
  Numerically speaking, our scheme has an advantage in all threshold scenarios.
  For the $(k,k)$-threshold, the theoretical contrast values of all schemes in Table \ref{tab:theoretical contrast comparisons} are uniformly $(\frac{1}{2})^{k-1}$. 
  For other $(k,n)$-thresholds where $n>k$, the contrast of our scheme is optimal. 
  Compared to other mainstream schemes, 
  our advantage lies in significantly increasing the probability of $\mathcal{K}$ 
  appearing in the generated $n$ bits, 
  while eliminating the need for a final global random arrangement operation. 
  This not only enhances the contrast but also makes the process more efficient and straightforward. 

  $\bullet$ 
  Analysis of Figure \ref{fig:Ours_Shyu_Yan} reveals that 
  our scheme exhibits a layering performance 
  when different shadow image combinations are selected for the recovery phase. 
 In contrast, the recovery effects of Shyu's and Yan's schemes remain unaffected by the selection of participant combinations. 
  As illustrated in Figure \ref{fig:Ours_Shyu_Yan}, 
  our scheme demonstrates superior reconstruction quality 
  at the first and the second level compared to the other two approaches, 
  and shows comparatively weaker results at the third level.
  It can be seen that our scheme is also suitable for scenarios where 
  different participants achieve varying effects in the recovery phase.

\begin{figure*}[t]
  \centering
  \includegraphics[width=0.9\linewidth, height=0.3\textheight, keepaspectratio=false]{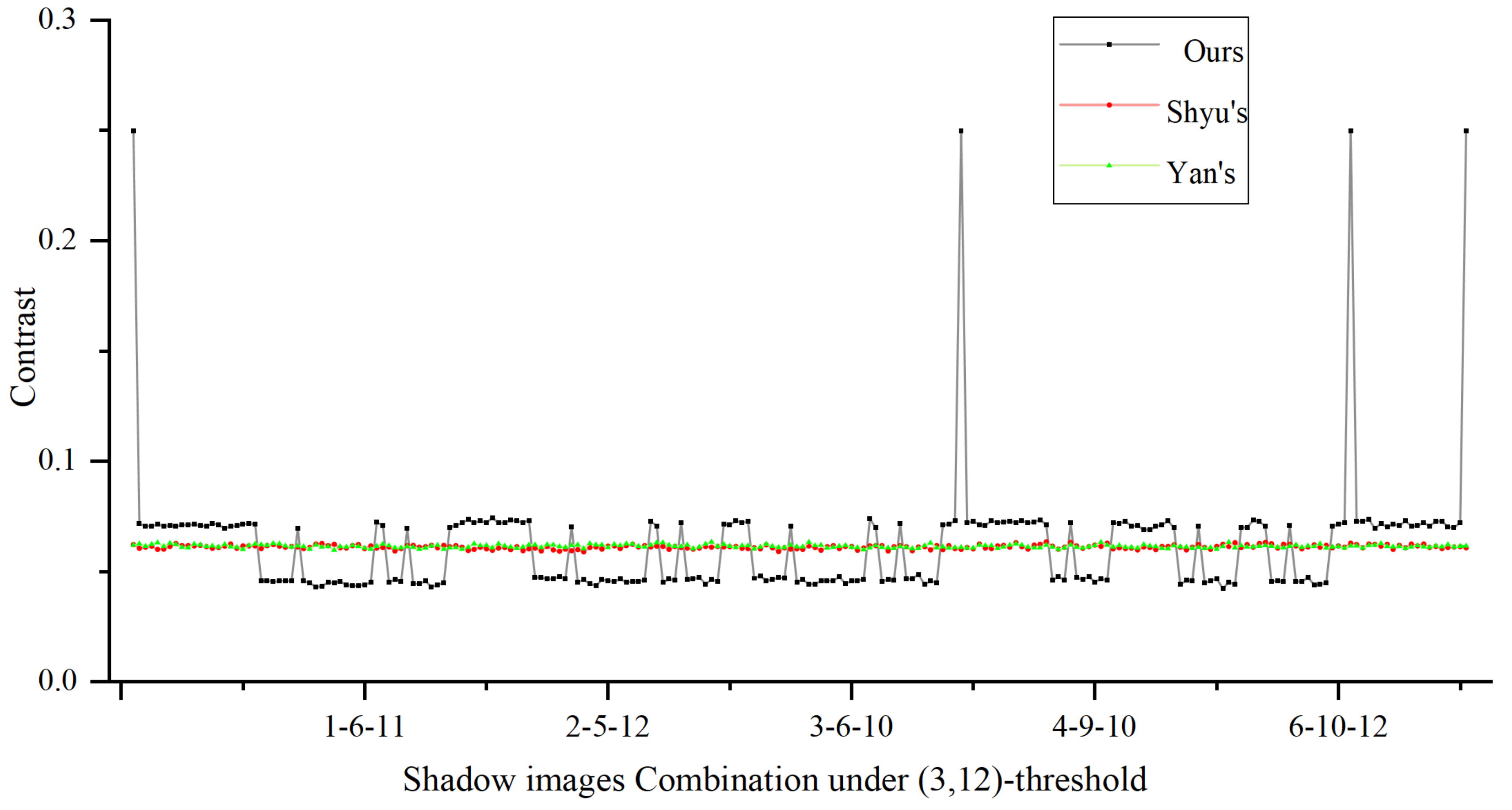}
  \caption{The contrast of all participant combinations in different schemes under $(3,12)$-threshold}
  \label{fig:Ours_Shyu_Yan}
  \vspace{-8pt} 
\end{figure*}

\section{Conclusion}\label{section:conclusion}
In this paper, 
we propose a novel sharing paradigm in RGVCS 
and introduce a higher-contrast $(k,n)$ RGVCS based on the paradigm. 
First, we demonstrate the layering contrast phenomenon under the new paradigm 
and provide a brand-new formula for contrast calculation. 
Second, we rigorously derive the contrast of the proposed scheme, 
and both theoretical analysis and experimental data confirm that 
our scheme achieves the current optimal contrast value. 
Future work will further explore 
the impact of the length of the initial bit group $n^{\prime}$ on the scheme 
and how to optimize our scheme further to achieve higher contrast.

\begin{table*}[]
  \centering
  \caption{The theoretical contrast comparisons with other $(k,n)$ schemes}\label{tab:theoretical contrast comparisons}
  \resizebox{\linewidth}{!}{
  \begin{tabular}{ccccccccccccccclcccc}
  \hline
        & \multicolumn{4}{c}{Chen \& Tsao} &  & \multicolumn{4}{c}{Wu \& Sun} &  & \multicolumn{4}{c}{Shyu} &  & \multicolumn{4}{c}{Ours}                                  \\ \cline{2-5} \cline{7-10} \cline{12-15} \cline{17-20} 
        & t=2    & t=3    & t=4    & t=5   &  & t=2   & t=3   & t=4   & t=5   &  & t=2 & t=3  & t=4  & t=5  &  & t=2            & t=3             & t=4              & t=5 \\ \hline
  (2,2) & 1/2    &        &        &       &  & 1/2   &       &       &       &  & 1/2 &      &      &      &  & 1/2            &                 &                  &     \\
  (2,3) & 1/7    & 1/4    &        &       &  & 2/7   & 1/2   &       &       &  & 2/7 & 1/2  &      &      &  & \textbf{3/10}  & 1/2             &                  &     \\
  (2,4) & 2/29   & 2/17   & 1/8    &       &  & 1/5   & 1/3   & 1/2   &       &  & 2/7 & 1/2  & 1/2  &      &  & \textbf{3/10}  & 1/2             & 1/2              &     \\
  (2,5) & 2/49   & 2/29   & 3/41   & 1/16  &  & 2/13  & 6/23  & 4/11  & 1/2   &  & 1/4 & 3/7  & 1/2  & 1/2  &  & \textbf{13/50} & \textbf{13/30}  & 1/2              & 1/2 \\
  (3,3) &        & 1/4    &        &       &  &       & 1/4   &       &       &  &     & 1/4  &      &      &  &                & 1/4             &                  &     \\
  (3,4) &        & 2/35   & 1/8    &       &  &       & 1/9   & 1/4   &       &  &     & 1/9  & 1/4  &      &  &                & \textbf{13/112} & 1/4              &     \\
  (3,5) &        & 1/44   & 4/83   & 1/16  &  &       & 1/16  & 3/22  & 1/4   &  &     & 2/23 & 4/21 & 1/4  &  &                & \textbf{5/56}   & \textbf{5/26}    & 1/4 \\
  (4,4) &        &        & 1/8    &       &  &       &       & 1/8   &       &  &     &      & 1/8  &      &  &                &                 & 1/8              &     \\
  (4,5) &        &        & 2/43   & 1/16  &  &       &       & 2/43  & 1/8   &  &     &      & 2/43 & 1/8  &  &                &                 & \textbf{67/1400} & 1/8 \\
  (5,5) &        &        &        & 1/16  &  &       &       &       & 1/16  &  &     &      &      & 1/16 &  &                &                 &                  & 1/16\\ \hline
  \end{tabular}
  }
  \end{table*}

\begin{table*}
	\caption{The experimental contrast comparisons with other $(k,n)$ schemes}\label{tab:experimental contrast comparisons}
	\centering
	\resizebox{\linewidth}{!}{
		\begin{tabular}{ccccccccccccccc}
			\hline
			& \multicolumn{4}{c}{Yan}           &  & \multicolumn{4}{c}{Shyu}          &  & \multicolumn{4}{c}{Ours}                                     \\ \cline{2-5} \cline{7-10} \cline{12-15} 
			& t=2    & t=3    & t=4    & t=5    &  & t=2    & t=3    & t=4    & t=5    &  & t=2             & t=3             & t=4             & t=5    \\ \hline
			(2,2) & 0.4989 &        &        &        &  & 0.5013 &        &        &        &  & 0.5000          &                 &                 &        \\
			(2,3) & 0.2888 & 0.5027 &        &        &  & 0.2884 & 0.5018 &        &        &  & \textbf{0.3039} & 0.5021          &                 &        \\
			(2,4) & 0.2868 & 0.5013 & 0.5013 &        &  & 0.2825 & 0.4962 & 0.4962 &        &  & \textbf{0.3005} & 0.5005          & 0.5005          &        \\
			(2,5) & 0.2492 & 0.4281 & 0.5001 & 0.5001 &  & 0.2488 & 0.4269 & 0.4979 & 0.4979 &  & \textbf{0.2594} & \textbf{0.4321} & 0.4980          & 0.4980 \\
			(3,3) &        & 0.2493 &        &        &  &        & 0.2494 &        &        &  &                 & 0.2476          &                 &        \\
			(3,4) &        & 0.1100 & 0.2499 &        &  &        & 0.1126 & 0.2513 &        &  &                 & \textbf{0.1172} & 0.2511          &        \\
			(3,5) &        & 0.0871 & 0.1897 & 0.2485 &  &        & 0.0856 & 0.1896 & 0.2500 &  &                 & \textbf{0.0918} & \textbf{0.1949}          & 0.2526 \\
			(4,4) &        &        & 0.1249 &        &  &        &        & 0.1253 &        &  &                 &                 & 0.1241          &        \\
			(4,5) &        &        & 0.0464 & 0.1250 &  &        &        & 0.0451 & 0.1239 &  &                 &                 & \textbf{0.0508} & 0.1269 \\
			(5,5) &        &        &        & 0.0608 &  &        &        &        & 0.0616 &  &                 &                 &                 & 0.0621 \\ \hline
		\end{tabular}
	}
	\vspace{-10pt} 
\end{table*}

\bibliography{reference}
\bibliographystyle{IEEEtran}
\end{document}